\documentclass[manuscript,screen]{acmart}

\usepackage{algorithm}
\usepackage{algorithmic}
\usepackage{multirow}
\usepackage{adjustbox,makecell}
\usepackage{subcaption}
\usepackage{graphicx}
\usepackage{enumitem}
\providecommand{\NewStructureName}[1]{}
\providecommand{\AssignStructureRole}[2]{}
\usepackage[most]{tcolorbox}

\AtBeginDocument{%
  }

\setcopyright{acmlicensed}
\copyrightyear{2018}
\acmYear{2018}
\acmDOI{XXXXXXX.XXXXXXX}
\acmConference[Conference acronym 'XX]{Make sure to enter the correct
  conference title from your rights confirmation email}{June 03--05,
  2018}{Woodstock, NY}

\acmISBN{978-1-4503-XXXX-X/2018/06}

\begin{document}

\title{FastFI: Enhancing API Call-Site Robustness in Microservice-Based Systems with Fault Injection}

\author{Yuzhen Tan}
\email{tanyuzhen@whu.edu.cn}
\orcid{0009-0001-0356-3568}
\affiliation{%
  \institution{School of Computer Science, Wuhan University}
  \city{Wuhan}
  \state{Hubei}
  \country{China}
}
\author{Jian Wang$^\ast$}
\email{jianwang@whu.edu.cn}
\orcid{0000-0002-1559-9314}
\affiliation{%
    \institution{School of Computer Science, Wuhan University, Zhongguancun Laboratory}
    \city{Wuhan}
    \state{Hubei}
    \country{China}
}
\author{Shuaiyu Xie}
\email{theory@whu.edu.cn}
\orcid{0000-0001-7925-3788}
\affiliation{%
    \institution{School of Computer Science, Wuhan University}
    \city{Wuhan}
    \state{Hubei}
    \country{China}
}
\author{Bing Li}
\email{bingli@whu.edu.cn}
\orcid{0000-0002-2165-2636}
\affiliation{%
    \institution{School of Computer Science, Wuhan University, Zhongguancun Laboratory}
    \city{Wuhan}
    \state{Hubei}
    \country{China}
}
\author{Yunqing Yong}
\email{2024202110044@whu.edu.cn}
\orcid{0009-0005-0575-7066}
\affiliation{%
    \institution{School of Computer Science, Wuhan University}
    \city{Wuhan}
    \state{Hubei}
    \country{China}
}
\author{Neng Zhang}
\email{nengzhang@ccnu.edu.cn}
\orcid{0000-0001-8662-5690}
\affiliation{%
    \institution{School of Computer Science \& Hubei Provincial Key Laboratory of Artificial Intelligence and Smart Learning, Central China Normal University}
    \city{Wuhan}
    \state{Hubei}
    \country{China}
}
\author{Shaolin Tan}
\email{shaolintan@hnu.edu.cn}
\orcid{0000-0001-6549-9760}
\affiliation{%
    \institution{Zhongguancun Laboratory}
    \city{Beijing}
    \country{China}
}

\renewcommand{\shortauthors}{Yuzhen Tan and Jian Wang, et al.}

\begin{abstract}
    Fault injection is a key technique for assessing software reliability, enabling proactive detection of system defects before they manifest in production. However, the increasing complexity of microservice architectures leads to exponential growth in the fault-injection space, rendering traditional random injection inefficient.  Recent lineage-driven approaches mitigate this problem through heuristic pruning, but they face two limitations. First, combinatorial-fault discovery remains bottlenecked by general-purpose SAT solvers, which fail to exploit the monotone and low-overlap structure of derived CNF formulas and typically rely on a static upper bound on fault size. Second, existing techniques provide limited post-injection guidance beyond reporting detected faults.
    To address these challenges, we propose FastFI, a fault-injection-guided framework to enhance the robustness of API call sites in microservice-based systems.
    FastFI features a DFS-based solver with dynamic fault injection to discover all valid combinatorial faults, and it leverages fault-injection results to identify critical APIs whose call sites should be hardened for robustness.
    Experiments on four representative microservice benchmarks show that FastFI reduces end-to-end fault-injection time by an average of 76.12\% compared to state-of-the-art baselines while maintaining acceptable resource overhead. Moreover, FastFI accurately identifies high-impact APIs and provides actionable guidance for call-site hardening.
\end{abstract}

\ccsdesc[500]{Software and its engineering~Software maintenance tools}

\keywords{microservice-based system, fault injection, API call-site robustness}


\maketitle

\section{Introduction}

Microservice architecture, a dominant architectural paradigm for modern software, has seen rapid adoption and scale-up. By decomposing monoliths into highly cohesive, loosely coupled services that interact via API calls, microservices offer architectural flexibility, development agility, and independent deployment~\cite{msintro1,msintro2}. However, as business requirements evolve and functionality becomes increasingly complex, microservice systems rapidly grow in both scale and heterogeneity: inter-service dependencies form intricate network topologies, and failure-propagation paths become difficult to predict or control. Meanwhile, the widespread adoption of DevOps (integration and automation of development and operations) further amplifies these pressures: continuous integration and continuous deployment compress release cycles to an unprecedented fast-paced rhythm~\cite{Devopsintro1,Devopsintro2}, outstripping traditional quality-assurance practices and making it increasingly challenging to detect and localize reliability defects in a timely manner.

Fault injection is a proactive reliability technique that deliberately injects controlled faults into otherwise healthy systems to observe behavior under adverse conditions. By deliberately injecting controlled faults into otherwise healthy systems, it provides actionable evidence for remediation and supports more accurate robustness assessment. Consequently, fault injection has been broadly adopted in industrial microservice deployments by leading practitioners such as Netflix~\cite{LDFI-Netflix,Netflix-ICSE}, Google~\cite{gcp-dirt-2020}, Microsoft~\cite{azure-chaos-studio}, and Amazon~\cite{aws-pg-chaos-2025,aws-pg-strategy-2025}.
Despite its effectiveness, applying fault injection at scale in microservice environments remains highly challenging. Unlike monolithic systems, microservice failures often emerge from complex interactions among multiple components, making it difficult to anticipate which fault scenarios are most critical to test. 

In practice, standard fault injection tools such as ChaosMesh~\cite{ChaosMesh}, ChaosBlade~\cite{ChaosBlade}, and ChaosMonkey~\cite{ChaosMonkey} support multi-dimensional injections spanning network anomalies, container failures, and node crashes. However, as the number of services and dependency paths grows, the fault space expands combinatorially, preventing these tools from producing actionable injection guidance within a practical timeframe. 
In many cases, injecting a single fault does not trigger an observable anomaly (e.g., due to redundancy and alternative execution paths), necessitating the injection of fault combinations. As illustrated in Fig.~\ref{fig:intro_combinatorial_fault}, injecting a fault into Cache or Database in isolation does not cause the request to fail; however, injecting faults into Cache and Database simultaneously leads to a request failure. We refer to such failure-inducing fault combinations and their corresponding injections as combinatorial faults and combinatorial-fault injection, respectively.  
A key objective is to identify a minimal injection that causes the system to fail, referred to as a minimal combinatorial fault. To address this problem,  Netflix~\cite{LDFI-Netflix} adapted LDFI~\cite{LDFI-2015} to microservices by modeling minimal combinatorial-fault discovery as enumerating minimal SAT solutions of a CNF (Conjunctive Normal Form) formula. IntelliFI~\cite{IntelliFI} and MicroFI~\cite{MicroFI} further employed heuristic pruning strategies to reduce the combinatorial-fault space and prioritize candidates based on fitness and API priority, thereby enabling the prioritized injection of high-value faults under limited injection budgets. Although prior work has made notable progress in pruning the fault injection space and prioritizing injections, two practical challenges remain prominent:

\begin{figure}[t]
  \centering
  \begin{subfigure}[t]{0.59\linewidth}
    \centering
    \includegraphics[width=\linewidth]{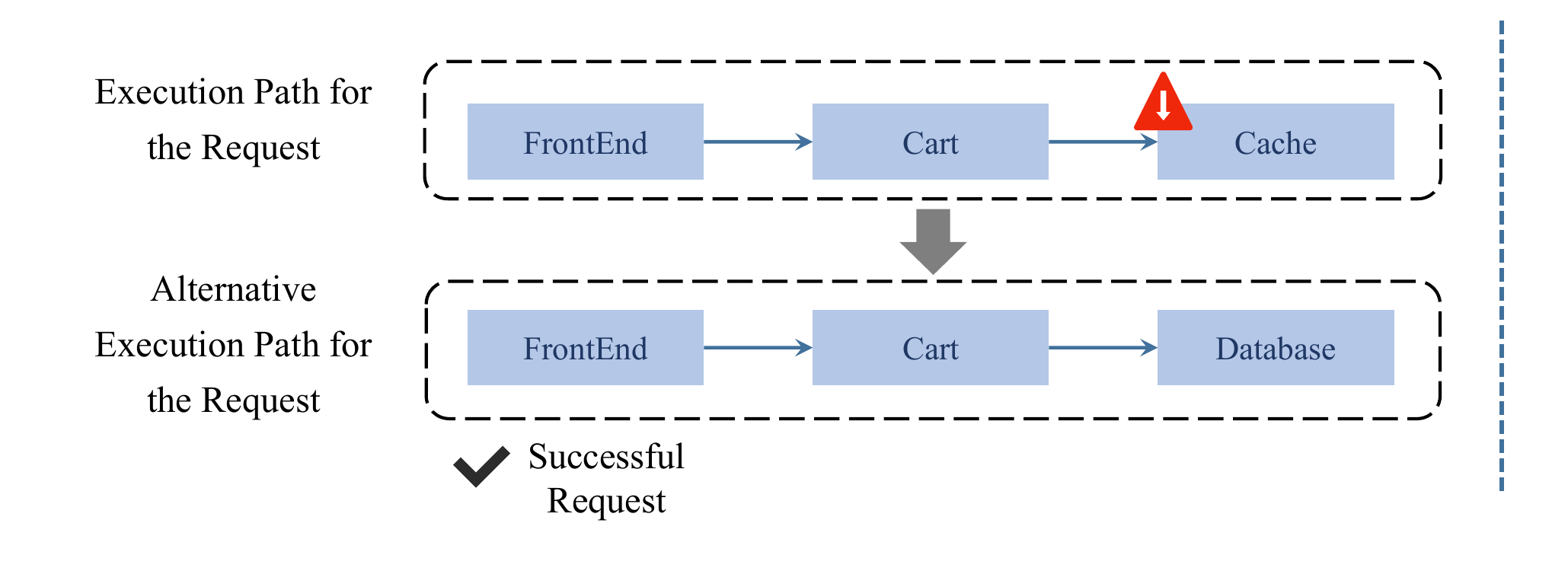}
    \caption{Single-fault injection $\{Cache\}$}
    \label{fig:intro_combinatorial_fault_a}
  \end{subfigure}\hfill
  \begin{subfigure}[t]{0.41\linewidth}
    \centering
    \includegraphics[width=\linewidth]{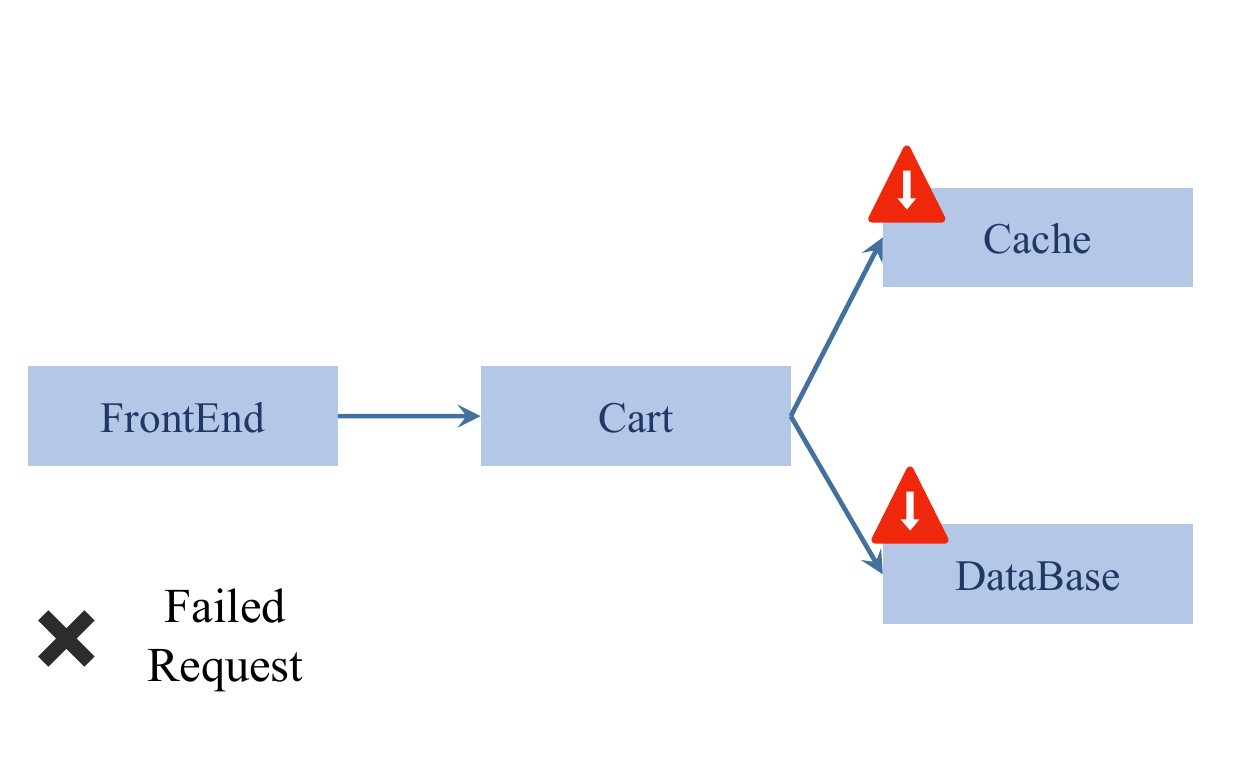}
    \caption{Combinatorial-fault injection $\{Cache, Database\}$}
    \label{fig:intro_combinatorial_fault_b}
  \end{subfigure}
  \caption{Illustration of fault injection.}
  \label{fig:intro_combinatorial_fault}
\end{figure}

\begin{figure*}[t]
  \centering
  \captionsetup{font=small}

  \begin{tabular}{@{}p{0.33\textwidth}@{\hspace{0.03\textwidth}}p{0.64\textwidth}@{}}

  \begin{minipage}[t]{\linewidth}\rule{0pt}{0pt}
    \centering
    \includegraphics[width=\linewidth,trim=10 10 10 10,clip]{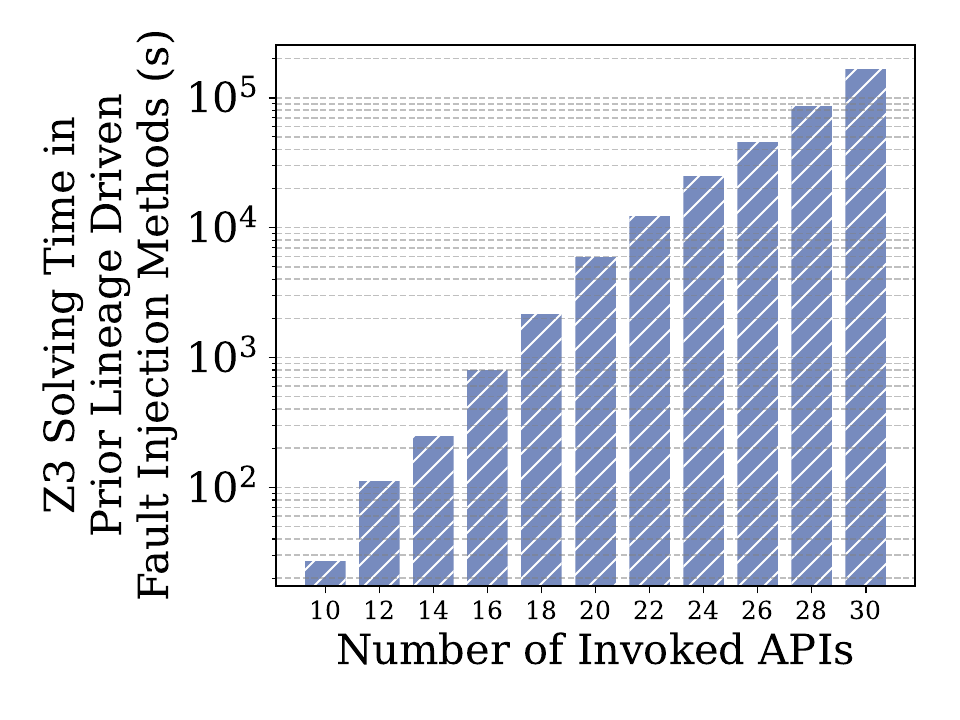}
    \captionof{figure}{Solving time for combinatorial faults.}
    \label{fig:intro_solving_time}
  \end{minipage}
  &
  \begin{minipage}[t]{\linewidth}\rule{0pt}{0pt}
    \centering
    \scriptsize
    \setlength{\tabcolsep}{2.6pt}
    \renewcommand{\arraystretch}{1.05}

    \captionof{table}{Statistics of Alibaba microservice trace dataset}
    \label{tab:Alibaba-Dataset-Result}

    \subcaptionbox{API invocation distribution per service request%
      \label{tab:Alibaba-Dataset-Result-servicce-api}}[0.49\linewidth]{%
      \centering
      \begin{tabular}{@{}c|c|c@{}}
        \toprule
        \makecell{\textbf{Number of}\\\textbf{Invoked APIs}} &
        \makecell{\textbf{Number of}\\\textbf{Service Requests}} &
        \makecell{\textbf{Proportion (\%)}} \\
        \midrule
        $[1,10]$       & 28281 & 82.72 \\
        $[11,20]$      & 3315  & 9.70  \\
        $[21,30]$      & 1095  & 3.20  \\
        $[31,50]$      & 831   & 2.43  \\
        $[51,+\infty)$ & 666   & 1.95  \\
        \bottomrule
      \end{tabular}
    }\hfill
    \subcaptionbox{Top-K APIs ranked by the number of covered service requests%
      \label{tab:Alibaba-Dataset-Result-api-service}}[0.49\linewidth]{%
      \centering
      \begin{tabular}{@{}c|c|c|c@{}}
        \toprule
        \textbf{K} &
        \textbf{API ID} &
        \makecell{\textbf{Covered Service}\\\textbf{Requests}} &
        \makecell{\textbf{Coverage}\\\textbf{Ratio (\%)}} \\
        \midrule
        1 & mJxUT-THwr  & 8352 & 24.43 \\
        2 & 7aJ\_JD61MG & 4801 & 14.04 \\
        3 & vs2nQhH1hq  & 2554 &  7.47 \\
        4 & g9oxz20ShF  & 2482 &  7.26 \\
        5 & vs2nQhH1hq  & 2328 &  6.81 \\
        \bottomrule
      \end{tabular}
    }
  \end{minipage}

  \end{tabular}
\end{figure*}

\textbf{(1) Efficiency bottleneck in discovering minimal combinatorial faults.} Prior approaches, such as LDFI, IntelliFI, and MicroFI, reduce the minimal combinatorial-fault discovery problem to enumerating minimal solutions of a CNF formula, yet they generally overlook solver-side efficiency. In microservice systems, high availability is often achieved by provisioning multiple alternative execution paths that overlap only weakly~\cite{low_overlap1, low_overlap2, low_overlap3} (see \S\ref{sec:motivation2}), thereby inducing CNF instances with low clause overlap. Moreover, the CNF formulas derived from LDFI are monotone, meaning that each clause contains only positive literals, i.e., no variable appears negated. Nevertheless, existing methods that rely on Z3~\cite{z3-solver}, a general-purpose SAT solver, do not exploit these structural features, making CNF solving the dominant runtime bottleneck in the fault-injection pipeline.
Based on our statistics on the Alibaba microservice trace dataset~\cite{alibaba-dataset} (Table~\ref{tab:Alibaba-Dataset-Result-servicce-api}), we find that although most requests invoke only a small number of API calls, 1,497 distinct requests within one hour still invoke more than 30 API calls. To illustrate the resulting computational burden, we consider a request with four alternative execution paths, where each path invokes 30 APIs. Under the evaluation environment described in Section~\ref{evaluation_environment}, solving the corresponding CNF formula can exceed 164,661 seconds, as shown in Fig.~\ref{fig:intro_solving_time}. At scale, repeatedly incurring day-level (or longer) solving time per request would lead to prohibitive end-to-end costs, making solver efficiency a first-order design requirement for combinatorial-fault discovery in large-scale systems.

Furthermore, these methods set an a priori upper bound on the number of APIs in each combinatorial fault, rather than dynamically adapting it according to the results of fault injection (see \S\ref{sec:motivation3}). A bound that is too small may prevent discovering valid combinatorial faults, whereas an overly large bound introduces substantial redundancy and wastes injection budget.

\textbf{(2) Lack of targeted robustness enhancement guidance.} Existing techniques emphasize how to design and execute injections, but provide limited post-injection guidance for localizing critical APIs and translating findings into actionable hardening decisions. As pointed out by the classical fallacies of distributed computing~\cite{dc_fallacies}, transient network unreliability, latency variability, and topology changes are inherent to distributed deployments and often lie beyond the control of any single service. Consequently, callee-side fixes alone are often insufficient to ensure timely and successful responses under all failure modes. Many practical mitigations for such failures (e.g., graceful degradation and proactive controls, as well as failover mechanisms along the request path) therefore need to be enforced at the caller side. We thus target API call sites, i.e., caller-owned program points where outbound requests are constructed and dispatched, and harden them with caller-side (call-site) policies to better withstand and adapt to these failures.


In large-scale microservice systems, however, there can be thousands of call sites to harden; thus, effective post-injection guidance is needed to prioritize those with the largest potential blast radius across requests. Based on the Alibaba microservice trace dataset~\cite{alibaba-dataset}, Table~\ref{tab:Alibaba-Dataset-Result-api-service} shows the distribution of API call counts per request and the prevalence of API sharing across endpoints in production. In the sample data, the most commonly used API is invoked by 8,352 service endpoints, accounting for 24.43\% of the total. This high degree of sharing implies that the robustness of call sites associated with a small set of highly shared APIs critically influences system stability: defects in error handling at these call sites can propagate across numerous business scenarios, resulting in systemic degradation of service quality. Therefore, accurately identifying and prioritizing improvements to such call sites based on fault-injection evidence is pivotal for microservice reliability assurance, especially under DevOps where continuous delivery further compresses remediation cycles.

To address these challenges, we propose FastFI, a framework that integrates dynamic fault injection with robustness enhancement guidance for API call sites. 
Firstly, FastFI exploits the monotone and weakly overlapping structure of CNF formulas induced by alternative execution paths in microservice systems and introduces a depth-first search (DFS) based solver with bitmask optimizations to enumerate minimal SAT solutions (see \S\ref{sec:method-FastSATSolver}); secondly, FastFI adaptively adjusts the upper bound of combinatorial-fault size, which corresponds to the number of injection points included in each combinatorial fault, based on real-time feedback during injection (see \S\ref{sec:method-K-Fault}); and finally, FastFI formulates critical API call-site identification as a Partial Max-SAT problem~\cite{PartialMaxSAT} (see \S\ref{sec:method-api-choose}), providing developers with actionable guidance for enhancing the most vulnerable API call sites. The main contributions of this paper are summarized as follows:

(1) We develop a DFS-based solver with bitmask-based optimizations that exploits the weakly overlapping structure and monotone CNF formulas induced by alternative microservice execution paths, along with a feedback-driven mechanism that dynamically adjusts the combinatorial-fault size during injection.


(2) We propose a post-injection robustness–enhancing method that identifies and prioritizes API call sites with critical impacts on reliability, providing developers with targeted system improvement suggestions.


(3) We evaluate FastFI on four representative microservice benchmarks and validate scalability on production-scale call graphs. Experimental results show that FastFI substantially outperforms state-of-the-art baselines in terms of efficiency while yielding actionable guidance for robustness enhancement.

The remainder of this paper is organized as follows. Section~\ref{sec:2} introduces preliminaries for combinatorial-fault injection in microservice systems and motivates the need for efficient fault discovery and robustness enhancement at API call sites. Section~\ref{sec:3} describes the design of FastFI, including our DFS-based SAT solver, the dynamic fault-injection mechanism, and how FastFI identifies critical API call sites to guide call-site hardening. Section~\ref{sec:4} presents the evaluation methodology and experimental results. Section~\ref{sec:5} discusses limitations and threats to validity. Section~\ref{sec:6} reviews related work. Section~\ref{sec:7} concludes the paper and outlines directions for future work.

\section{Background and Motivation}
\label{sec:2}
\subsection{Background}
This work targets reliability validation and robustness improvement in microservice-based systems. In such systems, combinatorial-fault injection is often required to expose interaction-induced failures that cannot be revealed by single-fault testing. Meanwhile, improving robustness in practice typically relies on enforcing caller-side policies at concrete API call sites. To facilitate understanding of our problem setting and solution space, we summarize the necessary background concepts below.

\textbf{Chaos Engineering}. Chaos engineering is a systematic methodology for increasing confidence in system behavior under turbulent conditions by intentionally injecting faults into production environments~\cite{ChaosEngineeringBasic}. It begins by defining measurable steady-state metrics that characterize normal behavior, formulates the hypothesis that these metrics hold under both control and experimental conditions, injects realistic failures, and tests the hypothesis by detecting deviations from the steady state. In industry, chaos engineering practice emphasizes strict blast-radius control, risk-aware execution, and incremental automation toward continuous experimentation.

\begin{figure}[t]
  \centering
  \begin{minipage}[b]{0.48\textwidth}
    \centering
    \includegraphics[width=\linewidth]{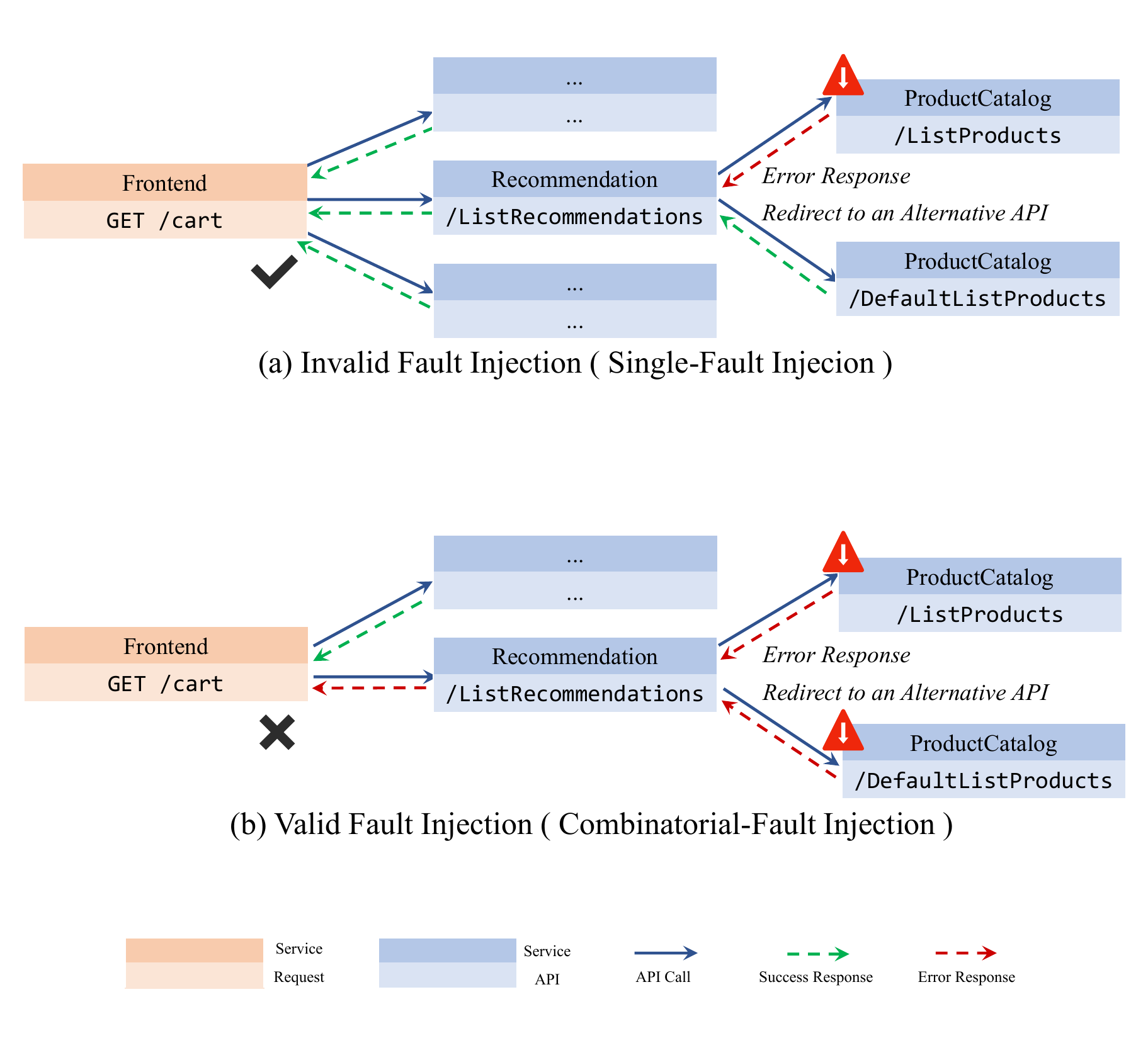}
    \captionof{figure}{End-to-end workflow of combinatorial-fault injection.}
    \label{fig:CombinationFault}
  \end{minipage}
  \hfill
  \begin{minipage}[b]{0.51\textwidth}
    \centering
    \includegraphics[width=\linewidth]{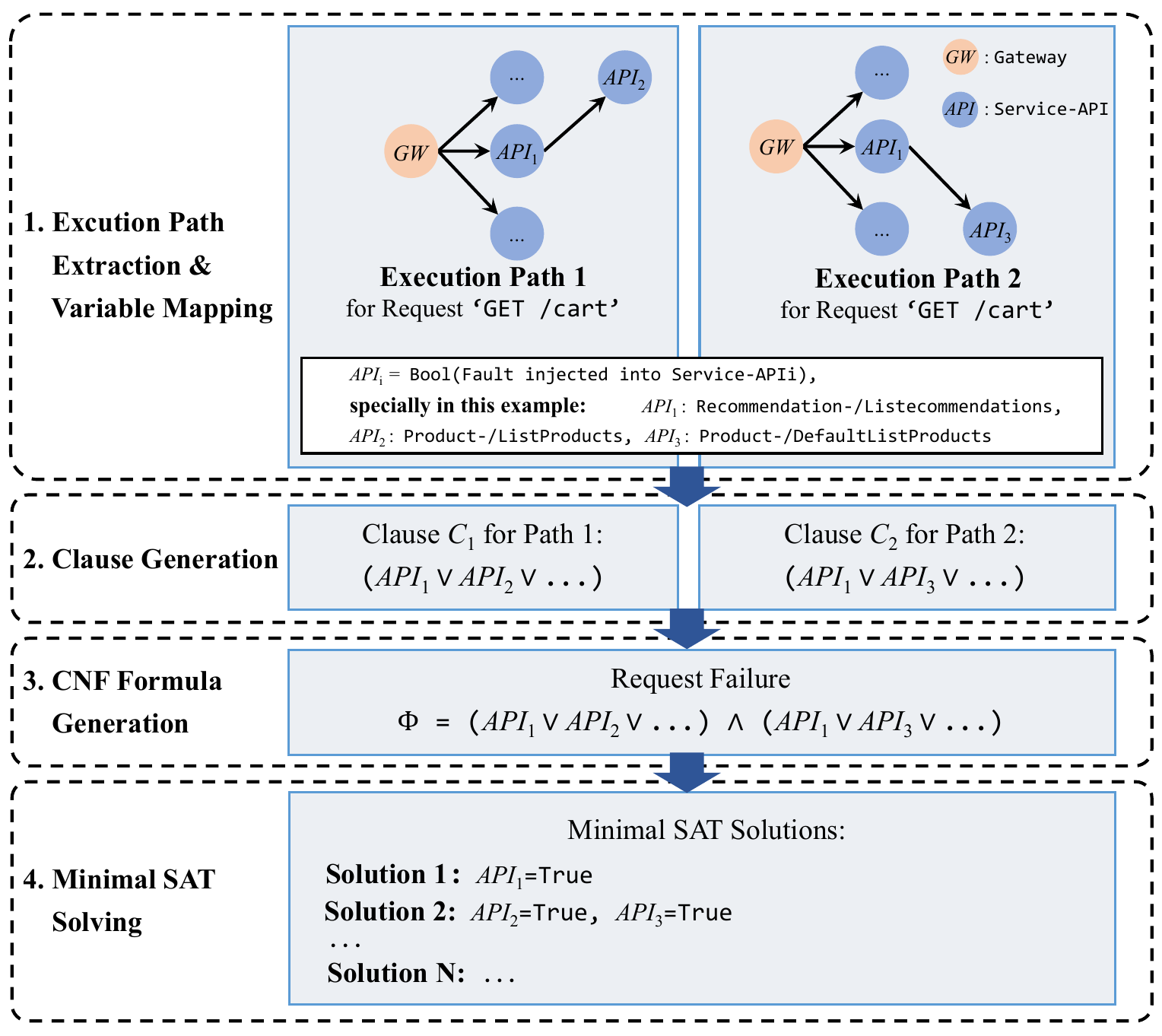}
    \captionof{figure}{Combinatorial-fault discovery pipeline in Online Boutique.}
    \label{fig:ldfi_process}
  \end{minipage}
\end{figure}

\begin{figure}[t]
  \centering
  \begin{minipage}[t]{0.51\linewidth}
    \centering
    \captionsetup{type=figure,width=\linewidth}%
    \rule{0pt}{0pt}
    \includegraphics[width=\linewidth]{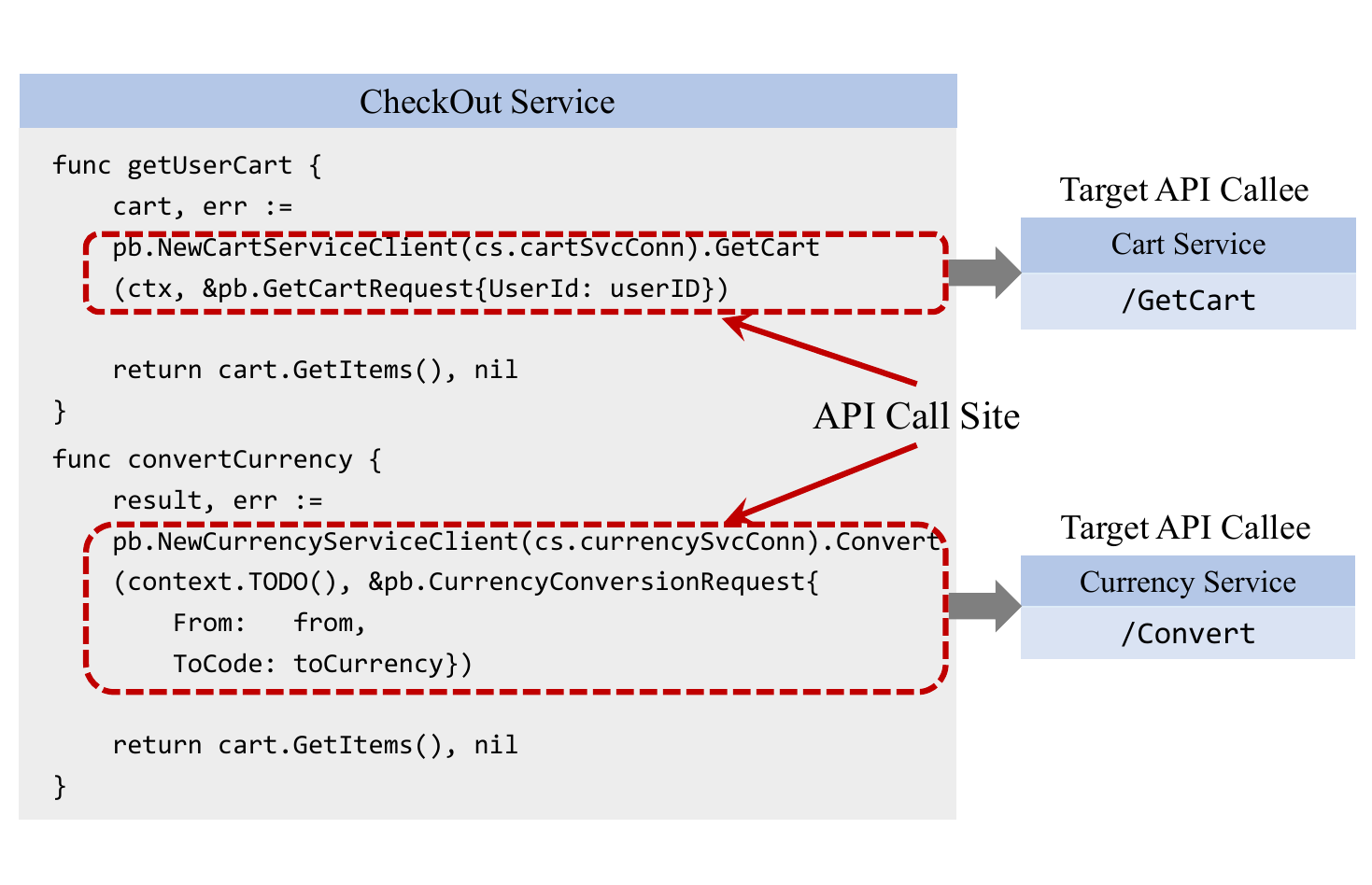}
    \captionof{figure}{Example API call sites in the source code.}
    \label{fig:api_call_site}
  \end{minipage}%
  \hfill
  \begin{minipage}[t]{0.47\linewidth}
    \centering
    \captionsetup{type=table,width=\linewidth}%
    \captionof{table}{Fault injection results in Online Boutique.}
    \label{tab:api-fault-injection-result}
    \rule{0pt}{0pt}
    \scriptsize
    \setlength{\tabcolsep}{2pt}
    \resizebox{\linewidth}{!}{%
      \begin{tabular}{@{}l|l|l|l@{}}
        \toprule
        \textbf{Request} & \textbf{Service} & \textbf{API} & \textbf{Result} \\
        \midrule
        \multirow{5}{*}{\texttt{GET /}}
          & cartservice           & \texttt{/GetCart}              & \textbf{500 Error} \\
          & CurrencyService       & \texttt{/Convert}              & \textbf{500 Error} \\
          & CurrencyService       & \texttt{/GetSupportedCurrency} & \textbf{500 Error} \\
          & productcatalogservice & \texttt{/ListProducts}         & \textbf{500 Error} \\
          & AdService             & \texttt{/GetAd}                & 200 OK    \\
        \midrule
        \multirow{7}{*}{\texttt{GET /product/\{id\}}}
          & CartService           & \texttt{/GetCart}              & \textbf{500 Error} \\
          & CurrencyService       & \texttt{/Convert}              & \textbf{500 Error} \\
          & CurrencyService       & \texttt{/GetSupportedCurrency} & \textbf{500 Error} \\
          & ProductCatalogService & \texttt{/GetProduct}           & \textbf{500 Error} \\
          & ProductCatalogService & \texttt{/ListProducts}         & \textbf{500 Error} \\
          & RecommendationService & \texttt{/ListRecommendations}  & \textbf{500 Error} \\
          & AdService             & \texttt{/GetAd}                & 200 OK    \\
        \midrule
        \multirow{7}{*}{\texttt{GET /cart}}
          & CartService           & \texttt{/GetCart}              & \textbf{500 Error} \\
          & CurrencyService       & \texttt{/Convert}              & \textbf{500 Error} \\
          & CurrencyService       & \texttt{/GetSupportedCurrency} & \textbf{500 Error} \\
          & ProductCatalogService & \texttt{/GetProduct}           & \textbf{500 Error} \\
          & ProductCatalogService & \texttt{/ListProducts}         & \textbf{500 Error} \\
          & RecommendationService & \texttt{/ListRecommendations}  & \textbf{500 Error} \\
          & ShippingService       & \texttt{/GetQuote}             & \textbf{500 Error} \\
        \bottomrule
      \end{tabular}%
    }
  \end{minipage}
\end{figure}

\textbf{Combinatorial-Fault Injection}. Combinatorial-fault injection tests system reliability by introducing multiple faults concurrently, thereby better reflecting production incidents that stem from co-occurring events such as hardware failures plus traffic surges. Fig.~\ref{fig:CombinationFault} illustrates combinatorial-fault injection in Online Boutique~\cite{OnlineBoutique}, a widely-used microservice benchmark, using the \texttt{GET /cart} request as a representative example. A fault in the ProductCatalog service does not immediately cause a failure because the call graph contains a fallback call to \texttt{/DefaultListProduct}. However, when faults are injected into both the primary and fallback APIs simultaneously, the request fails, revealing a vulnerability that single-fault tests do not expose. Consequently, systematic combinatorial testing uncovers interaction-induced weaknesses, enabling developers to accurately identify minimal combinatorial faults, thereby improving microservice reliability under complex failures. A combinatorial fault is considered \emph{valid} if injecting the corresponding set of faults triggers a request failure. Otherwise, it is deemed \emph{invalid} because it does not induce any observable request failure.

\textbf{Lineage-Driven Fault Injection}. To reduce the search space of candidate combinatorial faults, Lineage-Driven Fault Injection (LDFI) constructs a CNF formula from the execution lineage graph and enumerates minimal solutions to guide fault injection~\cite{LDFI-2015}. Netflix adapted LDFI and deployed it at scale in microservice-based systems~\cite{LDFI-Netflix}. In microservice-based systems, a request trace that includes alternative execution paths can be modeled as a lineage graph. Let $API_{i,j}$ be a Boolean variable indicating whether the $j$-th API on alternative path $i$ fails. In this setting, a combinatorial fault is \emph{valid} if every alternative path contains at least one failed API; it is \emph{invalid} if there exists at least one healthy path. This produces a CNF formula: $\Phi \;=\; \bigwedge_{i=1}^{m}\;(\bigvee_{j=1}^{k_i} API_{i,j})$, where $m$ is the number of alternative paths, $k_i$ is the number of APIs on path $i$, and each clause of the CNF formula is an alternative path. Notably, $\Phi$ is a monotone CNF formula, as it contains no negated variables (i.e., no $\lnot API_{i,j}$). Calculating the minimal SAT solutions of $\Phi$ gives the minimal combinatorial faults that cause the request to fail.

Fig.~\ref{fig:ldfi_process} illustrates the process of discovering combinatorial faults for the \texttt{GET /cart} request in Online Boutique. Specifically, we model each API call along an execution path as a Boolean variable $API_i$. If $API_i$ is set to \texttt{True}, a fault is injected at the corresponding API. For each execution path, we collect the set of Boolean variables representing the API calls on that path. We then construct one CNF clause per path by disjoining all variables on the path. Finally, we conjoin the per-path clauses to obtain the overall CNF formula for the request, as follows: $(API_1 \lor API_2 \lor \cdots)\ \land\ (API_1 \lor API_3 \lor \cdots)$. Solving this CNF formula and extracting its minimal satisfying assignments yields candidate fault-injection plans such as $\{API_1\}\ \text{and}\ \{API_2, API_3\}$.

\textbf{API Call Site}. An API call site is a specific program location (e.g., a function-call statement) where a service invokes an external API provided by another service or a third-party component. In microservice systems, each API call site represents a potential point of inter-service dependency and thus a candidate location for enhancing program robustness. Fig.~\ref{fig:api_call_site} shows the API call sites in the source code of Online Boutique’s Checkout Service. In particular, the statements enclosed by the red dashed boxes indicate the exact lines where the Checkout Service issues outbound requests to downstream services (e.g., the Cart Service and the Currency Service). We refer to each such program location as an API call site.

\textbf{Robustness of API Call Sites}. In microservice-based systems, we define the robustness of an API call site as the client’s ability to maintain availability despite downstream faults, timeouts, and request rejections. This robustness is typically realized through (i) failover, which routes traffic to healthy replicas when anomalies occur; (ii) graceful degradation, which keeps the system partially functional using cached data or default responses when replicas are down or severely degraded; and (iii) proactive controls, such as calibrated timeouts, load-aware retries, and circuit breakers, to prevent retry storms and mitigate cascading failures. Robust call sites reduce downtime, limit cascading dependency failures, and help services meet their SLOs in cloud-native deployments.

\subsection{Motivation}
\subsubsection{API Call Sites with Low Robustness Result in Low Service Reliability}\label{sec:motivation1}
While microservice architectures bring benefits such as agile development and independent deployment, enabled by service decoupling, they also introduce distributed-systems complexity and new reliability risks. In microservice-based systems, a single request often needs to traverse multiple services and invoke dozens of APIs to complete a business process. This high degree of distribution amplifies the impact of individual downstream failures, often expanding the blast radius substantially, particularly when API call sites lack adequate robustness mechanisms.

To assess the impact of the robustness of the API call site on service reliability, we injected gRPC status code 14 (UNAVAILABLE) faults into the Online Boutique benchmark and performed repeated trials for three representative requests: the homepage (\texttt{GET /}), the product-details page (\texttt{GET /product/\{id\}}), and the shopping-cart page (\texttt{GET /cart}). As shown in Table~\ref{tab:api-fault-injection-result}, except for AdService's \texttt{/GetAd} API, injecting faults into other critical downstream APIs caused the corresponding upstream endpoints to respond with HTTP 500 (Internal Server Error).

For most API calls, this behavior occurs because Online Boutique implements neither fallback alternatives nor degradation mechanisms, resulting in low robustness of API call sites. When a downstream service returns an error, the upstream service immediately aborts and propagates the failure, which can trigger cascading failures along the request chain. By contrast, the system continues to respond successfully even when calls to AdService fail. This is because the system employs a lenient handling strategy for API call results of AdService; even when its API call returns errors, upstream services do not interrupt the execution of the primary business logic. This design pattern indicates that AdService is defined as an optional dependency in the current business process, where its service availability does not substantially impact the normal operation of core business functionality.

\begin{tcolorbox}[
  enhanced,
  colback=gray!5,
  colframe=gray!500,
  arc=2mm,
  boxrule=1.5pt,
  left=1.5mm,right=1.5mm,top=1mm,bottom=1mm
]
\textbf{Insight 1:}
API call sites with inadequate robustness can trigger cascading failures, causing user-facing requests to return HTTP 500 (Internal Server Error) and reducing service availability. This observation underscores the need to systematically enhance robustness at API call sites to sustain reliable service operation at scale.
\end{tcolorbox}

\subsubsection{Efficiency of Discovering Valid Combinatorial Faults} \label{sec:motivation2}
In highly decoupled microservice architectures, a key design principle for improving API call-site robustness and end-to-end availability is to provide multiple fallback paths with minimal shared dependencies for each business function. As DoorDash reports~\cite{doordash-rpc-fallbacks-2022}, if each path succeeds with probability $0.9$ and failures are approximately independent, then two such paths can already yield a $99\%$ success rate (i.e., $1-(1-0.9)^2$). Conversely, when paths share critical dependencies such as a common database or a third-party API, a single-point failure can simultaneously invalidate all paths, eliminating the intended redundancy. Reflecting industry best practices, engineering guidelines from Oracle~\cite{oracle-cnc-policy-180} and IBM~\cite{ibm-sva-failover} recommend maintaining a hierarchical structure with three replicas of the same request and a primary path. This architectural pattern of the few-but-independent-path has been widely adopted in practice; however, it poses unique computational challenges for fault injection testing.

The LDFI algorithm abstracts each alternative execution path as a clause in a monotone CNF formula and discovers combinatorial faults by enumerating all minimal satisfying assignments of the resulting formula. In production systems, alternative execution paths are typically designed for high availability. To prevent failures on one path from impairing others, these paths are engineered to be largely independent with minimal overlap; i.e., the corresponding clauses share few variables.
This design choice is further motivated by the common-mode/correlated-failure phenomenon: shared components and hidden dependencies can induce correlated outages that simultaneously invalidate multiple paths, undermining the intended redundancy~\cite{low_overlap1, low_overlap2, low_overlap3}.
While this abstraction faithfully captures the structure of such fallback designs, it also yields challenging instances for SAT solvers. General-purpose SAT solvers are optimized for broad classes of problems; however, when clauses are near-independent and variable sharing is sparse, standard search pruning and conflict-learning heuristics can become substantially less effective, leading to poor solving efficiency.

\begin{figure}[t]
  \centering
  \includegraphics[width=\linewidth]{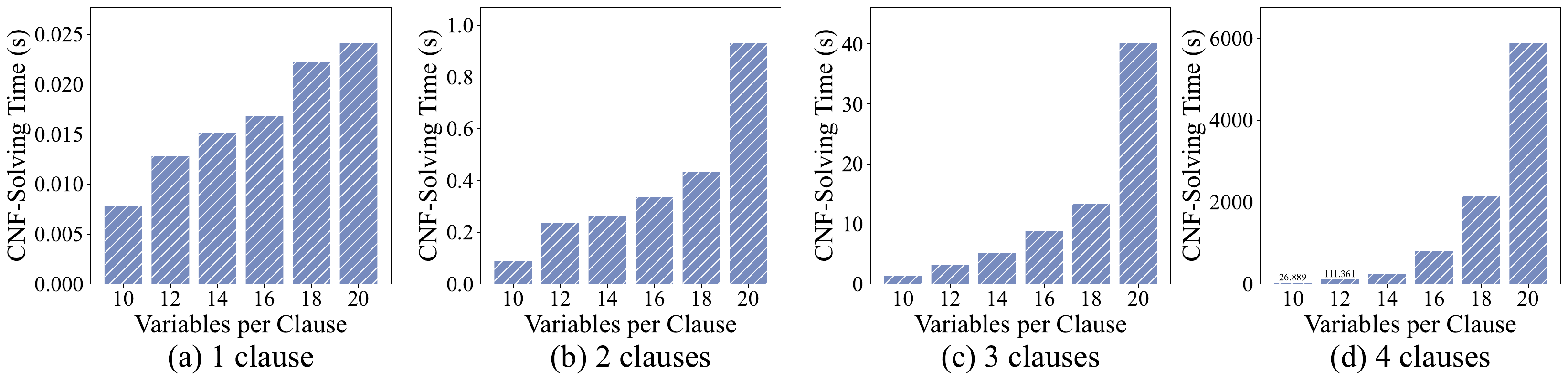}
  \caption{CNF-solving time of the Z3 solver for computing minimal satisfying assignments.}
  \label{fig:motivation-z3}
\end{figure}

When applying general-purpose SAT solvers such as Z3~\cite{z3-solver} to enumerate minimal SAT solutions, we observe exponential solving time growth even for modest instance sizes. Fig.~\ref{fig:motivation-z3} shows that with only four clauses, each containing 20 variables, the solving time reaches $5.8 \times 10^3$ seconds. Although the production trace statistics in Table~\ref{tab:Alibaba-Dataset-Result-servicce-api} indicate that most requests invoke only a small number of API calls, more than 2,500 requests invoke over 20 API calls, making this worst-case behavior practically relevant. Moreover, Z3 returns one minimal SAT solution at a time: after obtaining one minimal SAT solution, we add a blocking clause and restart the search (even with incremental solving), incurring substantial branching and backtracking in each round. More critically, as microservice systems scale up, this computational bottleneck amplifies further: when systems contain a large number of microservices, each maintaining multiple alternative paths and each alternative path invoking more API calls, the search space for minimal SAT solutions quickly becomes prohibitively large, thereby failing to meet DevOps timeliness requirements.

\begin{tcolorbox}[
  enhanced,
  colback=gray!5,
  colframe=gray!500,
  arc=2mm,
  boxrule=1.5pt,
  left=1.5mm,right=1.5mm,top=1mm,bottom=1mm
]
\textbf{Insight 2:}
General-purpose SAT solvers often struggle to scale on the few-but-independent-path structure common in microservice-based systems, where clauses share variables sparsely, and the formula is monotone. This highlights the urgent need to develop specialized structure-aware algorithms that exploit low clause overlap to support real-time fault injection testing in large-scale microservice environments.
\end{tcolorbox}

\subsubsection{Necessity of Dynamic Combinatorial-Fault Injection}\label{sec:motivation3}
Existing combinatorial-fault injection methods, such as LDFI~\cite{LDFI-Netflix}, IntelliFI~\cite{IntelliFI}, and MicroFI~\cite{MicroFI}, require configuring an upper bound on the number of APIs included in each combinatorial fault before initiating tests. This static bound implicitly assumes that system scale and fault-tolerance assumptions are known and stable. However, in continuously evolving real-world microservice environments, the number of services and the dependency topologies change across versions, rendering such a static choice brittle.

To assess the sensitivity of static bounds and avoid benchmark-specific conclusions, we selected one representative request from each of four widely used microservice benchmarks, including Online Boutique~\cite{OnlineBoutique}, Hotel Reservation~\cite{hotel}, Sock Shop~\cite{SockShop}, and Train Ticket~\cite{TrainTicket} (details of these benchmarks are provided in Section~\ref{sec: benchmarks}). These benchmarks cover diverse service topologies and dependency depths, allowing us to test whether the observed effects persist across distinct microservice call graphs rather than reflecting idiosyncrasies of a single benchmark.  The selected requests involve 7, 5, 9, and 15 API invocations, respectively. We then conducted static-bound fault-injection experiments on each request.

\begin{table}[t]
  \centering
  \caption{Results of static combinatorial-fault injection.}
  \label{tab:static_fault_num}
  \setlength{\tabcolsep}{2pt}
  \renewcommand{\arraystretch}{1.3}
  \resizebox{\linewidth}{!}{%
    \begin{tabular}{l|
      *{5}{>{\centering\arraybackslash}m{2.6em}}|
      *{5}{>{\centering\arraybackslash}m{2.6em}}|
      *{5}{>{\centering\arraybackslash}m{2.6em}}|
      *{3}{>{\centering\arraybackslash}m{2.6em}}
    }
      \toprule
      \textbf{Request} &
      \multicolumn{5}{c}{\makecell{\textbf{Online Boutique:} \\ \texttt{GET /cart}}} &
      \multicolumn{5}{|c|}{\makecell{\textbf{Hotel Reservation:} \\ \texttt{GET /hotels}}} &
      \multicolumn{5}{c|}{\makecell{\textbf{Sock Shop:} \\ \texttt{POST /orders}}} &
      \multicolumn{3}{c}{\makecell{\textbf{Train Ticket:} \\ \texttt{POST /left}}} \\
      \midrule
      \textbf{Max Number of APIs} &
        2 & 3 & 4 & 5 & 6 &
        2 & 3 & 4 & 5 & 6 &
        2 & 3 & 4 & 5 & 6 &
        2 & 3 & 4 \\
      \textbf{Number of Valid Faults} &
        0 & 0 & 7 & 7 & 7 &
        0 & 0 & 5 & 5 & 5 &
        0 & 0 & 9 & 9 & 9 &
        0 & 15 & 15 \\
      \textbf{Total Fault Injections} &
        68  & 338  & 1360 & 4507  & 18919 &
        26  & 92   & 361  & 988   & 2816  &
        129 & 708  & 3698 & 15448 & 67458 &
        361 & 3620 & 27873 \\
      \bottomrule
    \end{tabular}%
  }
\end{table}

In the evaluation, we varied the upper bound on the number of APIs contained in each combinatorial fault according to the replication factor of the benchmark. For the Train Ticket deployed with three replicas, we set the bound to range from 2 to 4. For the other three benchmarks, each deployed with four replicas, we set the bound to range from 2 to 6. As shown in Table~\ref{tab:static_fault_num}, when the preset quota is too conservative, potential combinatorial defects remain undiscovered; when the quota is too large, the search space expands exponentially, wasting testing resources and causing unnecessary disruption to business operations. Consequently, static methods either leave dangerous blind spots in evolving systems or waste computational resources, making them increasingly ill-suited for modern CI/CD pipelines.

\begin{tcolorbox}[
  enhanced,
  colback=gray!5,
  colframe=gray!500,
  arc=2mm,
  boxrule=1.5pt,
  left=1.5mm,right=1.5mm,top=1mm,bottom=1mm
]
\textbf{Insight 3:}
Configuring a static bound on the number of APIs per combinatorial fault ignores the evolving dynamics of microservice-based systems, causing insufficient exploration or excessive testing effort. Fault injection should adaptively adjust the bounds based on runtime feedback as the system evolves, balancing the efficiency of discovering valid combinatorial faults with testing cost.
\end{tcolorbox}

\section{Method}
\label{sec:3}
\subsection{Overview}
\begin{figure}[t]
  \centering
  \includegraphics[width=\linewidth]{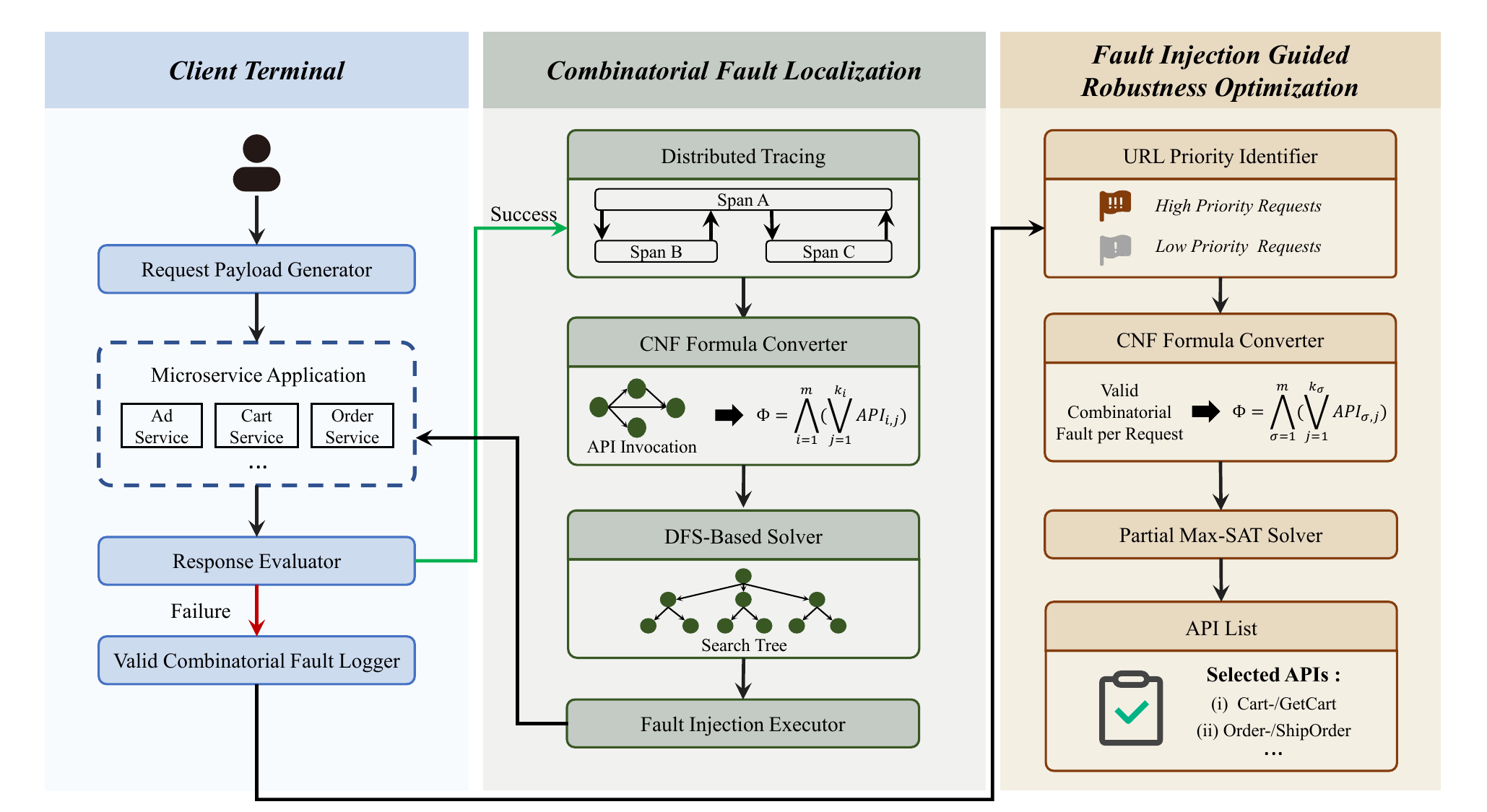}
  \caption{System architecture of FastFI.}
  \label{fig: FastFI-Overview}
\end{figure}
Fig.~\ref{fig: FastFI-Overview} illustrates FastFI, a framework designed to enhance the robustness of API call sites through fault injection. Guided by request-level combinatorial-fault injection and combining logical constraint solving with optimization techniques, FastFI aims to provide a low-cost and high-benefit solution for improving robustness in DevOps pipelines, while ensuring a controlled fault blast radius and maintaining online service availability. The framework consists of three core components:
\begin{itemize}[leftmargin=*,labelsep=0.5em]
    \item \textbf{Client Terminal.} This component first generates specific request payloads and issues them to the target microservice application. It then observes and evaluates the responses to determine whether the injected fault induces a failure. For failure-inducing cases, the client terminal records the corresponding fault as a valid combinatorial fault. For success cases, it forwards the execution outcome to the subsequent localization stage to guide further solving and discover additional combinatorial faults. Overall, the outcome of each trial (success or failure) is fed back to prune subsequent candidates, thereby efficiently filtering and validating combinatorial faults.
    \item \textbf{Combinatorial Fault Localization.} This component targets robustness assessment in microservice-based systems by identifying valid combinatorial faults that can trigger request failures under a realistic fault bound. Given a user request and a tester-specified maximum number of APIs in a combinatorial fault, our goal is to discover all minimal valid combinatorial faults whose injection is sufficient to cause the request to fail. Concretely, the component first collects the distributed trace associated with the request and encodes the execution path extracted from the trace into a monotone CNF formula. It then applies a DFS-based SAT solver to solve all size-bounded minimal satisfying assignments. Finally, the fault injection executor translates each identified combinatorial fault into an EnvoyFilter configuration and injects the corresponding faults into the Kubernetes-deployed application for validation. Unlike unbounded fault injection, which can always force a failure by injecting enough faults, FastFI focuses on the smallest and most critical combinations within a practical limit, thereby revealing which combinatorial faults are most likely to compromise the request and providing actionable guidance for enhancing API call-site robustness.
    \item \textbf{Fault Injection Guided Robustness Optimization.} This component performs robustness optimization by identifying critical APIs whose call-site robustness should be enhanced under a constrained budget. After enumerating all valid combinatorial faults for each request, it assigns requests to high-priority and low-priority tiers using a chosen prioritization policy (e.g., user access frequency). It then converts the per-request valid combinatorial faults into a monotone CNF formula that encodes which APIs must be hardened at their call sites to prevent request failures. Finally, it formulates a Partial Max-SAT problem that maximizes robustness gains while satisfying the budget constraints, and solves it to obtain a list of critical APIs. This list guides which call sites to enhance robustness (e.g., via failover, graceful degradation, or proactive controls).
\end{itemize}

\subsection{Combinatorial Fault Localization}
\subsubsection{DFS-Based Solver for Enumerating All Minimal SAT Solutions of the Monotone CNF Formula}\label{sec:method-FastSATSolver}

As discussed in Motivation~2 (Section~\ref{sec:motivation2}), Z3 becomes a bottleneck for enumerating all minimal SAT solutions of the monotone CNF formula constructed from alternative execution paths for a request instance, thereby reducing the efficiency of discovering valid combinatorial faults. To address this issue, we design a specialized solver for enumerating all minimal SAT solutions based on depth-first search (DFS) combined with bitmasking and pruning rules to accelerate the enumeration of minimal satisfying assignments for the monotone CNF formula. The solver exploits structural properties of the monotone CNF formula induced by modeling alternative execution paths, resulting in faster solving.
As shown in Algorithm~\ref{alg:dfs_based_solver}, our DFS-based solver for minimal SAT solutions comprises three components: \textbf{State Representation}, \textbf{Search Strategy}, and \textbf{Pruning Strategy}.

\renewcommand{\algorithmicrequire}{\textbf{Input:}}
\renewcommand{\algorithmicensure}{\textbf{Output:}}
\begin{algorithm}[t]
\caption{DFS-Based Solver for Enumerating All Minimal SAT Solutions}
\label{alg:dfs_based_solver}
\begin{algorithmic}[1]
\REQUIRE CNF formula $\Phi=\bigwedge_{i=1}^{m}\big(\bigvee_{j=1}^{k_i} API_{i,j}\big)$, upper bound of combinatorial-fault size $current\_crashes$
\ENSURE all minimal SAT solutions $minimal\_solutions$
\STATE $clause\_list \gets \mathrm{BuildClauseMasks}(\Phi)$ \COMMENT{ $clause\_list[i]$: list $(mask,node)$ of clause $i$ }
\STATE $best\_depth \gets \mathrm{InitializeDepthStateMap}()$
\STATE $full\_mask \gets 2^{m}-1$,\ $maxd \gets current\_crashes$
\STATE $stack \gets [(full\_mask,0,\emptyset)]$,\ $minimal\_solutions \gets [\,]$
\WHILE{$stack \neq \emptyset$}
  \STATE $(u,d,chosen) \gets pop(stack)$
  \IF{$u = 0$}
    \IF{{($chosen$ \textbf{in} $minimal\_solutions$)} \OR {(\textbf{not} $\mathrm{IsMinimal}(chosen, full\_mask)$)}} \STATE \textbf{continue} \ENDIF
    \STATE $minimal\_solutions.\textsc{append}(chosen)$,\ \textbf{continue}
  \ENDIF
  \IF{$d \ge maxd$ \OR $d > best\_depth[u]$} \STATE \textbf{continue} \ENDIF
  \STATE $best\_depth[u] \gets d$
  \STATE $i \gets bit\_length(u \land -u) - 1$
  \FORALL{$(mask,node)$ \textbf{in} $clause\_list[i]$}
    \IF{$node \in chosen$} \STATE \textbf{continue} \ENDIF
    \STATE $stack.\textsc{append}\big((u \land \sim mask,\ d+1,\ chosen \cup \{node\})\big)$
  \ENDFOR
\ENDWHILE
\RETURN $minimal\_solutions$
\end{algorithmic}
\end{algorithm}

\textbf{State Representation:}
To reduce memory overhead during the search, we represent the clause-coverage state of a candidate combinatorial fault using an $m$-bit mask $\mathit{mask}\in\{0,1\}^{m}$ (implemented as an integer), where $m$ is the number of clauses in the CNF formula
$\Phi=\bigwedge_{i=1}^{m}\left(\bigvee_{j=1}^{k_i} API_{i,j}\right)$.
Specifically, the $i$-th bit $C_i$ of $\mathit{mask}=(C_1,\ldots,C_m)$ encodes the coverage of clause $i$: $C_i=1$ denotes that clause $i$ is uncovered, and $C_i=0$ denotes that it is covered. We initialize $\mathit{mask}$ to all ones, i.e., $C_i=1$ for all $i\in\{1,\ldots,m\}$, meaning that no clause has been covered yet.

\textbf{Search Strategy:}
The solver performs a depth-first search (DFS) over the space of coverage masks. At each step (Lines 5--6), it scans $\mathit{mask}$ from the least significant bit upward and identifies the first uncovered clause (Line 17). It then selects a literal $API_{i,j}$ from clause $i$ (i.e., an injection site on the corresponding execution path), adds it to the current combinatorial fault $F \leftarrow F \cup \{API_{i,j}\}$, and updates the coverage mask by clearing the bits of all clauses covered by $API_{i,j}$:
\begin{equation}\label{eq:mask-update}
\mathit{mask} \leftarrow \mathit{mask}\ \land\ \sim\Bigl(\operatorname{bitmask}\bigl(\operatorname{cover}(API_{i,j})\bigr)\Bigr).
\end{equation}
In Eq.~\eqref{eq:mask-update}, $\operatorname{cover}(API_{i,j})$ denotes the set of clauses that would be satisfied once $API_{i,j}$ is selected. The function $\operatorname{bitmask}(\cdot)$ maps this set to a binary mask whose $t$-th bit is $1$ iff clause $t$ is covered by $API_{i,j}$. Taking the bitwise AND with its negation therefore clears precisely those covered clauses, producing a new state for subsequent search (Lines 18--23).

\textbf{Pruning Strategy:}
To prevent faults from exceeding the user-specified bound, the search prunes any branch whose depth exceeds the maximum number of APIs included in each combinatorial fault (Lines 13--15). In addition, we maintain a table $\mathit{best\_depth}[\mathit{mask}]$ that records the minimum depth at which each coverage mask has been reached. When the DFS revisits the same $\mathit{mask}$ at a larger depth, the branch is discarded because it cannot produce a solution with fewer selected APIs than an already explored state. Together, these rules substantially reduce redundant exploration while preserving completeness within the depth bound.

The DFS reaches a satisfying assignment when all $m$ bits of the mask become $0$ (i.e., $\mathit{mask}=0$), meaning that the current combinatorial fault $F$ covers all clauses in $\Phi$ (Line 7). However, such a satisfying assignment may still contain redundant APIs. Therefore, for every satisfying combinatorial fault found, we further apply an \emph{redundancy} (subset-minimality) check: we keep $F$ iff removing any single API from $F$ makes $\Phi$ unsatisfied (Lines 8--11). After this minimality filtering step, the remaining solutions constitute the final set of minimal fault-injection plans.

\begin{figure}[t]
  \centering
  \includegraphics[width=0.9\linewidth]{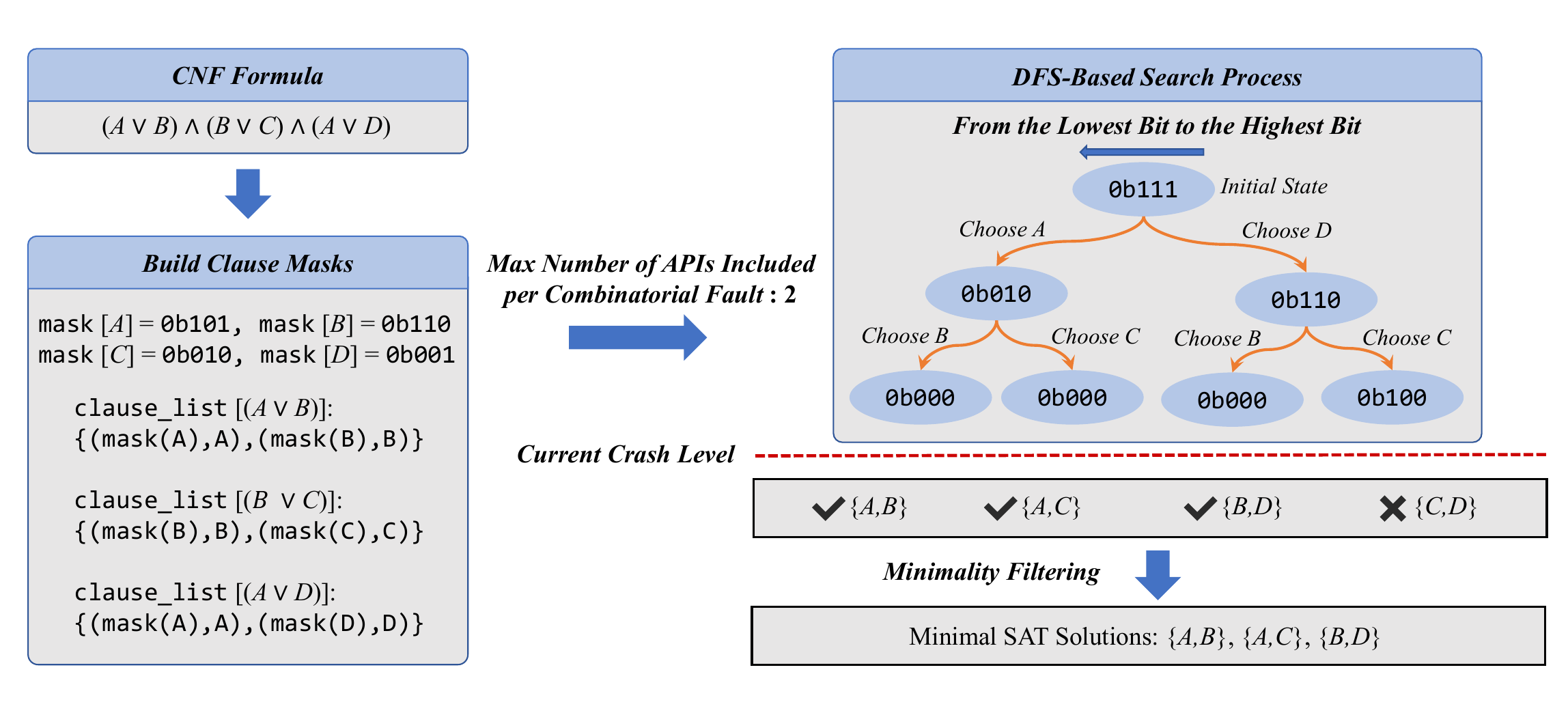}
  \caption{Illustrative example of the DFS-based solver for enumerating minimal SAT solutions.}
  \label{fig: FastDFS}
\end{figure}

Consider the CNF formula in Fig.~\ref{fig: FastDFS}, which demonstrates the step-by-step execution of our DFS-based solver for enumerating all minimal SAT solutions. The input formula is $(A \lor B)\ \land\ (B \lor C)\ \land\ (A \lor D)$, and the maximum number included per combinatorial fault is set to $k=2$. We first encode each clause as one bit in an uncovered-clause mask. Thus, the initial state is $\mathit{mask}=0\text{b}111$, meaning that all three clauses remain uncovered. Next, we compute the clause-coverage bitmask of each Boolean variable (API) and build the clause-to-candidate index $\mathit{clause\_list}$ as shown on the left of Fig.~\ref{fig: FastDFS}. For example, $A$ covers clauses $(A\lor B)$ and $(A\lor D)$, hence its coverage mask is $0\text{b}101$; similarly, $B$, $C$, and $D$ have coverage masks $0\text{b}110$, $0\text{b}010$, and $0\text{b}001$, respectively.

The DFS proceeds by always selecting the lowest-index uncovered clause, i.e., the least significant $1$-bit in the current mask. Starting from $0\text{b}111$, the solver first selects clause $(A \lor D)$ and branches on its candidates $A$ and $D$ (top branch on the right of Fig.~\ref{fig: FastDFS}). If it chooses $A$, the uncovered mask is updated by $\mathit{mask}\leftarrow \mathit{mask}\land \sim 0\text{b}101=0\text{b}010$, indicating that only clause $(B\lor C)$ remains uncovered. The solver then branches on $(B\lor C)$ and chooses either $B$ or $C$, yielding masks $0\text{b}000$ for both $\{A,B\}$ and $\{A,C\}$. Reaching $\mathit{mask}=0$ means all clauses are satisfied, so both sets are added to the candidate set of satisfying solutions.

In the other branch, if the solver chooses $D$ at the root, the uncovered mask is updated as $\mathit{mask}\leftarrow 0\text{b}111\land \sim 0\text{b}001=0\text{b}110$, which means that clauses $(A\lor B)$ and $(B\lor C)$ remain uncovered. The solver then continues by scanning from the least significant bit and selects the lowest-index uncovered clause, which is $(B\lor C)$. Accordingly, it branches on the candidates in that clause (i.e., $B$ or $C$): choosing $B$ updates the mask to $\mathit{mask}\leftarrow 0\text{b}110\land \sim 0\text{b}110=0\text{b}000$, yielding the satisfying set $\{B, D\}$; choosing $C$ updates the mask to $\mathit{mask}\leftarrow 0\text{b}110\land \sim 0\text{b}010=0\text{b}100$, which indicates that clause $(A\lor B)$ is still uncovered.

Finally, after obtaining all satisfying candidates up to depth $k=2$, the solver applies minimality filtering (redundancy): a candidate is kept only if removing any single variable makes the formula unsatisfiable. In this example, $\{A, B\}$, $\{A, C\}$, and $\{B, D\}$ all pass the check and therefore constitute the final minimal SAT solutions: $\{\{A, B\},\{A, C\},\{B, D\}\}$. In contrast, the candidate $\{C, D\}$ is pruned immediately upon generation: it reaches the maximum depth $k=2$ while still leaving the first clause $(A\lor B)$ uncovered, and thus cannot be extended into any satisfying solution. Consequently, that node and all of its descendants are discarded from further exploration.

Next, we analyze the algorithm’s time complexity as follows:
\begin{theorem}[Time Complexity]
\label{proof:time_complex}
\textnormal{Let $m$ be the number of clauses in the CNF formula, $V$ the maximum number of variables per clause, $k$ the maximum number of APIs included in each combinatorial fault, and $\bar s$ the average number of clauses covered by a single Boolean variable. Then, the time complexity of the DFS-based solver for enumerating all minimal SAT solutions is:}
\begin{equation}\label{eq:time-complexity-leafmin}
  T
  = O\Bigl(
    V\,(k+1)(k^2+1)\,\sum_{d=0}^{\lfloor m/\bar s\rfloor}\binom{m}{d\cdot\bar s}
  \Bigr),
\end{equation}

\textnormal{Furthermore, if the coverage rate $\bar s$ is sufficiently large, we obtain the stronger bound:}
\begin{equation}\label{eq:time-complexity-leafmin-another}
  T
  \ll O\Bigl(
    V\,(k+1)\,(k^2+1)\,2^m
  \Bigr).
\end{equation}
\end{theorem}

\begin{proof}
We prove the time complexity bound through five steps:
\paragraph{\textbf{(i) Analysis of State Expansions.}} The solver maintains a pruning table $\mathit{best\_depth}[u]$ for each uncovered-mask $u\in\{0,1\}^m$, initialized to $k+1$. A state $(u,d)$ is expanded only when $d \le \mathit{best\_depth}[u]$, after which the table is updated to $\mathit{best\_depth}[u]\leftarrow d$. Any subsequent visit to the same mask $u$ at depth $d'> \mathit{best\_depth}[u]$ is pruned without expansion. Let $R$ denote the number of distinct masks ever reached. Since $\mathit{best\_depth}[u]$ can decrease at most $k+1$ times (for $d\in\{0,1,\ldots,k\}$), each reachable mask is expanded at most $k+1$ times, and thus the total number of state expansions is bounded by $E \le (k+1)R$.
\paragraph{\textbf{(ii) Per-Expansion Computational Cost.}} In each expansion, the algorithm firstly selects one uncovered clause using constant-time bit operations. It then iterates over at most $V$ candidate variables in the selected clause; for each candidate, it performs $O(1)$ work to update the mask and push the resulting child state onto the DFS stack. Therefore, the dominant cost per expansion is the clause-candidate iteration, which costs $O(V)$. Across all expansions, the total branching cost is $O(VE)=O\bigl(V(k+1)R\bigr)$.
\paragraph{\textbf{(iii) Leaf Minimality Check Cost.}} Whenever the search reaches a satisfying leaf state ($u=0$), it invokes \textsc{IsMinimal} on the selected set $chosen$. In a direct implementation, \textsc{IsMinimal} removes each selected variable in turn and recomputes the remaining coverage to test necessity, which costs $O(k^2)$ time per call. Let $C$ denote the total number of satisfying leaf hits. The total cost of leaf minimality checks is thus $O(Ck^2)$. Moreover, the number of leaf hits is upper-bounded by the number of generated child states. Since each expansion generates at most $V$ children, we have $C \le VE$.
\paragraph{\textbf{(iv) Refined Bound on Reachable Masks.}} Let $\bar s$ denote the average number of clauses covered by selecting one variable. Consequently, reachable uncovered masks typically lie on Hamming-weight layers $\{m,\,m-\bar s,\,m-2\bar s,\,\ldots,\,m-\lfloor m/\bar s\rfloor\cdot \bar s\}$, which yields $R \le \sum_{d=0}^{\lfloor m/\bar s \rfloor}\binom{m}{d\cdot\bar s}$.
\paragraph{\textbf{(v) Overall Time Complexity.}} Combining (i)--(iv), we obtain $T = O\bigl(VE + Ck^2\bigr)$. Using $E \le (k+1)R$ and $C \le VE$, we further have $T = O\bigl(V(k+1)(k^2+1)R\bigr)$. Substituting the bound on $R$ from (iv) yields Eq.~\eqref{eq:time-complexity-leafmin}. Finally, when the coverage rate $\bar s$ is sufficiently large, we have $R \ll 2^m$, which implies Eq.~\eqref{eq:time-complexity-leafmin-another}.
\end{proof}

\subsubsection{Dynamic $k$-Fault Injection}\label{sec:method-K-Fault}

\begin{algorithm}[t]
\caption{$k$-Fault Based Dynamic Fault Injection}
\label{alg:kfault-dynamic-inject}
\begin{algorithmic}[1]
\REQUIRE URL request $url$, search bound $k_{\max}$
\ENSURE the set $\mathcal{V}$ of valid combinatorial faults
\STATE $\Phi \leftarrow \mathrm{GetCNFClause}(\mathrm{ExecuteWithNoFaultInject}(url))$
\STATE $k \leftarrow 1,\ \mathcal{V} \leftarrow \emptyset,\ stack \leftarrow \mathrm{DFSBasedSolver}(\Phi,k)$
\WHILE{$stack \neq \emptyset$}
    \STATE $S \leftarrow \mathrm{pop}(stack)$
    \IF{ ($\exists V_0 \in \mathcal{V}:\ V_0 \subseteq S$) \OR ($\mathrm{IsHistoryInjectedFault}(S)$)}
        \STATE \textbf{continue}
    \ENDIF
    \STATE $inject\_result \leftarrow \mathrm{ExecuteWithFaultInject}(url,S)$
    \STATE $path \leftarrow \mathrm{FindNewURLPath}(inject\_result)$
    \IF{$path = \mathrm{None}$}
        \STATE $\mathcal{V} \leftarrow \mathcal{V} \cup \{S\},\ \mathrm{RecordSolution}(S)$
        \STATE \textbf{continue}
    \ENDIF
    \STATE $\Phi \leftarrow \Phi \wedge \mathrm{GetCNFClause}(path)$
    \STATE $\mathrm{clear}(stack),\ stack \leftarrow \mathrm{DFSBasedSolver}(\Phi,k)$
    \IF{$stack = \emptyset\ \AND\ k < k_{\max}$}
        \STATE $k \leftarrow k+1,\ stack \leftarrow \mathrm{DFSBasedSolver}(\Phi,k)$
    \ENDIF
\ENDWHILE
\STATE \textbf{return} $\mathcal{V}$
\end{algorithmic}
\end{algorithm}

As discussed in Motivation 3 (Section~\ref{sec:motivation3}), static approaches fix the maximum number of APIs included in each combinatorial fault, which either fails to discover faults in the system or causes additional invalid fault injections. To address this problem, we propose a dynamic $k$-fault injection framework, as shown in Algorithm~\ref{alg:kfault-dynamic-inject}, which adaptively increases the combination size according to the fault injection feedback.

$k$-fault based dynamic fault injection takes a URL request $url$ and an upper bound $k_{\max}$ as input, and returns the set $\mathcal{V}$ of valid combinatorial faults. It first executes the request without fault injection, extracts the corresponding execution path, and encodes it into an initial CNF formula $\Phi$ (Line~1). We initialize the search with $k\!\leftarrow\!1$, an empty valid combinatorial-fault set $\mathcal{V}$, and a candidate stack with all minimal SAT solutions of $\Phi$ whose cardinality is $k$, as enumerated by our DFS-based solver (Line~2).

The main loop (Line~3) iteratively selects a candidate combinatorial fault $S$ from the stack (Line~4). To avoid unnecessary injections, we prune $S$ in two cases (Lines~5--7): (i) \emph{subsumption}, where there exists a previously valid fault $V_0 \in \mathcal{V}$ such that $V_0 \subseteq S$; because adding injection points cannot mask an induced failure, injecting $S$ would be redundant; and (ii) \emph{repetition}, where $S$ has been attempted before and recorded as invalid, since re-injecting the same non-failure-inducing fault is uninformative. Otherwise, we inject the faults in $S$ into the Kubernetes-deployed application via Istio's EnvoyFilter~\cite{envoyfilter} and evaluate the impact on the request (Lines~8--9). If the injection causes the request to fail, i.e., no alternative execution path can successfully serve the request and the observed failure rate exceeds a predefined threshold, we deem $S$ a valid combinatorial fault, add it to $\mathcal{V}$, and record it (Lines~10--13).

If the injection does not lead to a failure, we treat the observed execution path as a newly discovered alternative path. We then encode the path into a clause by $\mathrm{GetCNFClause}(path)$ and update the CNF formula by conjunction (Line~14). Since updating $\Phi$ may invalidate previously feasible solutions---i.e., a minimal SAT solution under the old formula may no longer satisfy the updated formula. We discard the existing candidate pool and re-solve the updated formula to prune away obsolete candidates, retaining only the minimal SAT solutions of the latest formula (Line~15). If the resulting candidate set is empty and $k < k_{\max}$, we increase the fault bound $k \leftarrow k+1$ and re-solve on the same updated formula (Lines~16--18). In this way, the procedure adaptively expands the search space from $k$-fault to $(k+1)$-fault combinations only when necessary, thereby reducing invalid injections while preserving completeness.

\subsection{Fault Injection Guided Robustness Optimization}\label{sec:method-api-choose}
The increasingly short DevOps iteration cycles~\cite{Devopsintro1,Devopsintro2} pose acute challenges to strengthening the robustness of API call sites. In each iteration, development teams can devote only limited time and engineering budget to enhancing call-site robustness. Meanwhile, fault injection experiments reveal valid combinatorial faults, leading to an optimization problem: under tight time and resource constraints, how can we quickly identify and prioritize API call sites whose robustness enhancement delivers the greatest impact at the lowest cost?

To prioritize effectively, we need to determine the importance of requests, which can be configured either statically or dynamically. In this paper, we use a frequency-driven dynamic prioritization mechanism. High-frequency requests often correspond to critical business workflows; their stability directly affects user experience and service continuity. Consequently, during API selection, we first prioritize mitigating combinatorial faults for these critical requests. For lower-priority requests, we adopt a coverage-maximization strategy to maximize the overall fault tolerance of the system to the greatest extent possible under a budget $B$, i.e., we can enhance the call-site robustness of at most $B$ APIs.

For a request $Req_i$, through the combinatorial-fault injection evaluation, we identify $m_i$ valid combinatorial faults $\{F_1, F_2, \ldots, F_{m_i}\}$, where each combinatorial fault $F_\sigma$ consists of $k_{i,\sigma}$ APIs. We define each API in a combinatorial fault as a decision variable $API_{\sigma,j}$, which indicates whether the corresponding API has been selected for enhancing its call-site robustness (assigned value 1 if chosen, and 0 otherwise).

Next, we construct the valid combinatorial faults for $Req_i$ as a CNF formula. Specifically, each fault $F_\sigma$ is transformed into a clause $\phi_{i,\sigma}$ that is the disjunction of all APIs in $F_\sigma$. The clause $\phi_{i,\sigma}$ is satisfied iff at least one API in $F_\sigma$ is chosen to enhance its call-site robustness by failover, graceful degradation, or proactive controls. This indicates that the corresponding fault has been effectively mitigated.

By conjoining all such clauses for $Req_i$, we obtain the overall CNF formula:
\begin{equation}
\Phi(Req_i) = \bigwedge_{\sigma=1}^{m_i} \left(\bigvee_{j=1}^{k_{i,\sigma}} API_{\sigma,j}\right) = \bigwedge_{\sigma=1}^{m_i} \phi_{i,\sigma}.
\end{equation}
Intuitively, $\Phi(Req_i)$ being satisfied means that all currently known valid combinatorial faults for $Req_i$ are mitigated. Equivalently, each combinatorial fault is mitigated by hardening the call-site robustness of at least one API in it. As a result, if the same combinatorial fault is injected again, the system will either degrade gracefully or reroute execution through an alternative path, thereby preventing user-visible failures under the injected fault.

Based on our business-importance analysis, the system must fully cover all known combinatorial faults for high-priority requests. Let $HPR$ and $LPR$ denote the sets of high-priority and low-priority requests, respectively. We require $\Phi(Req_i)$ to be satisfied for all $Req_i \in HPR$ and maximize coverage over $LPR$ under a budget $B$ that limits the number of APIs whose call-site robustness can be enhanced. The optimization problem is formulated as follows:

\begin{equation}
\begin{aligned}
\max\quad
& \sum_{Req_{Low}\in LPR}\mkern4mu
  \sum_{\phi_{Low,\sigma}\in\Phi(Req_{Low})} \phi_{Low,\sigma}, \\
\text{s.t.}\quad
& \Phi(Req_{High})=\text{True},\hspace{1em} \forall\, Req_{High}\in HPR, \\
& \sum_{j=1}^{n} API_j \le B,  \hspace{1em} \qquad API_j\in\{0,1\}.
\end{aligned}
\end{equation}

The API identification problem naturally maps to a Partial Max-SAT problem. The clauses of the CNF formula corresponding to high-priority requests are treated as \emph{hard} clauses, ensuring robustness for critical business requests, while the clauses of low-priority requests are treated as \emph{soft} clauses, allowing for maximizing coverage under limited enhancement resources. In this way, the problem of selecting APIs to enhance their call-site robustness reduces to a standard Partial Max-SAT optimization process.

We address this problem using a two-stage approach. In the first stage, for the hard clauses of high-priority requests $\Phi_{\text{Hard}}$, we apply the DFS-based solver for enumerating all minimal SAT solutions described in Section~\ref{sec:method-FastSATSolver} to enumerate minimal solution sets $\{\mathcal{H}^{(1)}, \mathcal{H}^{(2)}, \ldots, \mathcal{H}^{(n)}\}$. For each minimal solution $\mathcal{H}^{(i)}$, we evaluate its coverage in the soft clauses of low-priority requests $\Phi_{\text{Soft}}$ as:

\begin{equation}
\begin{aligned}
\text{cov}(\mathcal{H}^{(i)}) = |\{\phi_s \in \Phi_{\text{Soft}} \mid \phi_s(\mathcal{H}^{(i)}) = 1\}|.
\end{aligned}
\end{equation}

We then select the minimal solution with maximum coverage, $\mathcal{H}^* = \arg\max_i \text{cov}(\mathcal{H}^{(i)})$, and set the associated API variables to true. If the budget $B$ is sufficient to satisfy all hard clauses, we then pass the residual budget and the remaining unsatisfied soft clauses to a MaxSAT solver~\cite{or-tools} to continue optimization, subject to the constraint $\sum_{j=1}^{n} API_j \leq B$. If the budget $B$ is insufficient to satisfy all hard clauses, the instance is infeasible and fails to solve. This two-stage process begins with the optimal minimal SAT solution that guarantees coverage for high-priority requests and then progressively maximizes fault coverage for lower-priority requests under the given budget.

\section{Evaluation}
\label{sec:4}
In this section, we conducted evaluations to answer the following research questions:
\begin{itemize} [leftmargin=*,labelsep=0.5em]
    \item \textbf{RQ1:} How effective and efficient is FastFI in discovering valid combinatorial faults?
    \item \textbf{RQ2:} How does FastFI's DFS-based solver perform on the monotone CNF formulas produced by our model?
    \item \textbf{RQ3:} To what extent does the dynamic fault injection mechanism improve the execution efficiency of combinatorial-fault injection?
    \item \textbf{RQ4:} How accurately can FastFI identify API call sites that should be hardened, and how effective is the resulting hardening?
    \item \textbf{RQ5:} Is FastFI's resource overhead acceptable in practical deployment settings?
\end{itemize}

For \textbf{RQ1}, we examine FastFI's effectiveness and computational efficiency in discovering valid combinatorial faults.
For \textbf{RQ2}, we measure the CNF-solving time of our DFS-based solver.
For \textbf{RQ3}, we evaluate the effectiveness of our dynamic injection mechanism in reducing the number of fault-injection trials and improving overall efficiency.
For \textbf{RQ4}, we assess the accuracy of FastFI in prioritizing vulnerable API call sites and evaluate the effectiveness of the hardening guided by FastFI.
For \textbf{RQ5}, we measure the resource overhead of deploying FastFI, thereby validating its feasibility in practice.

\subsection{Evaluation Design}
\subsubsection{Benchmarks}
\label{sec: benchmarks}

\begin{table}[t]
\centering
\caption{Statistics of execution spans and invoked APIs for selected requests across benchmarks.}
\label{tab:span_api_stats}
\scriptsize
\setlength{\tabcolsep}{3pt}
\renewcommand{\arraystretch}{1.3}
\makebox[\linewidth][c]{
\resizebox{\linewidth}{!}{
\begin{tabular}{c|l|c|c||c|l|c|c}
\toprule
Benchmark & Request & \#Spans & \#Invoked APIs &
Benchmark & Request & \#Spans & \#Invoked APIs \\
\midrule
\multirow{2}{*}{Online Boutique}   & GET /                & 28  & 5  &
\multirow{2}{*}{Online Boutique}   & GET /product/\{id\}  & 26  & 7  \\
                                  & GET /cart            & 28  & 7  &
                                  & POST /cart/checkout  & 39  & 12 \\
\midrule
\multirow{3}{*}{Hotel Reservation} & GET /recommendations & 5   & 2  &
\multirow{3}{*}{Sock Shop}         & GET /login           & 15  & 2  \\
                                  & POST /reservation    & 5   & 2  &
                                  & POST /cart           & 23  & 2  \\
                                  & GET /hotels          & 13  & 5  &
                                  & POST /orders         & 34  & 9  \\
\midrule
\multirow{3}{*}{Train Ticket}      & POST /left           & 171 & 15 &
\multirow{3}{*}{Train Ticket}      & POST /cheapest       & 317 & 19 \\
                                  & POST /minStation     & 279 & 21 &
                                  & POST /quickest       & 317 & 19 \\
                                  & GET /food            & 25  & 6  &
                                  & POST /preserve       & 261 & 25 \\
\bottomrule
\end{tabular}
}%
}%
\end{table}

\mbox{}\par
\textbf{(1) Real-World Microservice Benchmarks.}  We evaluated FastFI on four widely used microservice benchmarks: Online Boutique~\cite{OnlineBoutique}, Sock Shop~\cite{SockShop}, Hotel Reservation~\cite{hotel}, and Train Ticket~\cite{TrainTicket}. Online Boutique adopts gRPC for inter-service communication, and implements distributed tracing with OpenTelemetry~\cite{OpenTelemetry}. Sock Shop and Train Ticket communicate via RESTful HTTP, and employ OpenTracing~\cite{OpenTracing} for distributed tracing. Hotel Reservation adopts gRPC for inter-service communication and also implements distributed tracing with OpenTracing. These applications vary in communication patterns, allowing us to validate the applicability and effectiveness of FastFI across diverse microservice architectures. As summarized in Table~\ref{tab:span_api_stats}, we selected, for each benchmark, requests that invoke the largest possible subset of APIs exposed by the business-logic services.

These benchmarks employ simplified designs without explicit fault-handling mechanisms, resulting in requests that lack alternative execution paths. As a result, a single component failure can cause a request to fail, limiting the scope of combinatorial fault studies. To address this limitation, following IntelliFI~\cite{IntelliFI} and MicroFI~\cite{MicroFI}, we constructed alternative execution paths by deploying multiple replicas for each microservice. We used Istio~\cite{istio} to enable automatic retries and traffic management so that when an instance fails, subsequent requests are automatically rerouted to healthy replicas. In our experiments, the Train Ticket benchmark was deployed with two and three replicas, and each of the other benchmarks was deployed with four and six replicas. For brevity in subsequent discussions, we denoted the benchmark configurations with four replicas as OB-4, HR-4, and SS-4 for Online Boutique, Hotel Reservation, and Sock Shop. Similarly, we used OB-6, HR-6, and SS-6 for their six replica configurations. We denoted the two-replica and three-replica Train Ticket configurations as TT-2 and TT-3, respectively.



\textbf{(2) Simulated Microservice Benchmark.} The above benchmarks are orders of magnitude smaller than production deployments, making them insufficient to stress-test algorithmic efficiency at scale. Moreover, the sets of APIs invoked by different requests are largely disjoint. This abstraction diverges from production deployments, where cross-service dependencies induce substantial overlap and correlation among requests, as shown in Table~\ref{tab:Alibaba-Dataset-Result-api-service}. To assess scalability and practicality, we built a simulated microservice benchmark comprising one million services across three types: business logic, database, and cache, where services of the business logic type account for the largest number. Following the call graph generation algorithm in~\cite{call-graph-generator}, which is derived from real-world production data, we applied a conditional probability model parameterized by the current depth and sibling node distribution. We capped per-step expansion to bound the maximum depth and to avoid unrealistically deep chains or excessive fan-out, thereby ensuring architectural plausibility.

To reflect production-style multipath, we generated alternative execution paths using a \emph{grouped-skeleton} scheme: for each request, backup paths are partitioned into groups; each group contains two paths that share a subset of edges termed the \emph{skeleton}, while all non-skeleton edges within the group are completely distinct. This mirrors multi-region deployments where distinct data centers or regions maintain independent call chains, enabling seamless failover and region-level isolation. To model caching behavior, we introduced intra-group asymmetry: one path models a cache hit with a lightweight fast path and fewer API calls, and the other models a cache miss with a full path that includes database access and more API calls. The two paths have an API-call ratio of \(3{:}7\). We exposed three tunable parameters for scaling: \emph{GroupNum} (the number of alternative groups per request), \emph{EdgeNum} (the total number of call edges across the two paths within a group), and \emph{BoneNum} (the number of shared edges that form the skeleton).

\subsubsection{Evaluation Environment and Implementation Details}
\label{evaluation_environment}
We built an experimental cluster on Kubernetes v1.29.15 with Istio v1.24.1 as the service mesh infrastructure. The cluster consists of six nodes: one control-plane node and five worker nodes. The control-plane node uses an Intel Xeon E3-1240 v5 (3.50\, GHz) with 16\, GB RAM, running Ubuntu 22.04 LTS. Each worker (Work-Node1 to Work-Node5) uses an Intel Core i7-4790 (3.60\, GHz) with 16\, GB RAM, also running Ubuntu 22.04 LTS.
FastFI is implemented in Python 3.9.22, and the source code is publicly available~\cite{FastFI}. For fault-injection experiments, depending on the protocol, we used EnvoyFilter to inject failures at the inbound path of the Envoy sidecar: gRPC status code 14 (UNAVAILABLE) for gRPC, and HTTP 503 (Service Unavailable) for RESTful HTTP. Meanwhile, retry behavior was configured centrally via Istio VirtualService. To perform request-level fault injection, we followed MicroFI’s strategy by generating a unique identifier as the hash of the request path and operation. We then embedded the hashing result in the tracestate~\cite{W3C-context} header field to scope fault propagation precisely. Since tracestate propagation is specified by the W3C context specification and supported by the OpenTelemetry standard~\cite{OpenTelemetry}, we modified the distributed tracing implementations in the Sock Shop, Hotel Reservation, and Train Ticket benchmarks to adopt OpenTelemetry.

\subsection{RQ1: Effectiveness of FastFI in Discovering Valid Combinatorial Faults}
\subsubsection{Baselines}

We evaluated FastFI against representative state-of-the-art lineage-driven fault-injection methods on real-world microservice benchmarks:
\begin{itemize}[leftmargin=*,labelsep=0.5em]
    \item \textbf{LDFI}~\cite{LDFI-Netflix}. It is an early lineage-driven technique adapted to microservices. It encodes alternative execution paths as a monotone CNF formula and enumerates minimal satisfying assignments, without relying on heuristic pruning.
    \item \textbf{IntelliFI}~\cite{IntelliFI}. It exploits recorded traces to prune the candidate space using rule-based heuristics. It refrains from injecting an upstream fault when the trace already manifests an upstream failure, and further removes candidates that are strict supersets of previously non-injected candidates.
    \item \textbf{MicroFI}~\cite{MicroFI}. It applies log-driven heuristics to eliminate candidates that have already been tested. By consulting the injection log, it reuses prior outcomes and avoids redundant injections.
\end{itemize}

\textbf{Request-Scoped History Reuse.} The key assumption underpinning MicroFI's reuse strategy, namely that the fault-handling logic for the same downstream API remains consistent across different call sites, does not always hold in practice. In the Online Boutique benchmark, injecting a fault into the recommendation service’s \texttt{/Listrecommendations} API causes \texttt{GET /product/\{id\}} and \texttt{GET /cart} to fail, but does not cause \texttt{POST /cart/checkout} to fail. This divergence arises from caller-dependent error handling and fallback logic, which can result in different end-to-end outcomes under the same downstream failure. Consequently, reusing API-level injection results across different requests can be unreliable and lead to erroneous fault injection conclusions. To ensure soundness in our evaluation, we applied history-based pruning conservatively and reused injection history only within the same request (i.e., per endpoint), rather than across different requests that invoke the same downstream API.

\textbf{Standardizing Baselines with Dynamic Configuration.} To ensure a fair comparison, we standardized all baselines by replacing their static configuration---a fixed combination size (i.e., a fixed maximum number of APIs included in each combinatorial fault) with a dynamic configuration, while leaving each method’s core logic and heuristics unchanged. Concretely, each baseline starts by injecting single-API faults and then incrementally increases the combination size, guided by online fault-injection feedback, until it discovers all valid combinatorial faults or reaches the configured maximum combination size. This standardization ensures that the evaluation primarily reflects end-to-end solving efficiency and the effectiveness of each method's pruning strategy.
Moreover, with a static configuration, the number of intermediate hypotheses can increase significantly on these benchmarks, resulting in an unwieldy number of candidate solutions. Since these baselines use Z3-based incremental solving, they must repeatedly add blocking constraints for previously enumerated models and re-solve to obtain the next candidate, which quickly becomes the dominant bottleneck as the candidate space grows. In our measurements on OB-4, for the request \texttt{GET /}, the static setting already becomes intractable when the CNF instance grows to six clauses: the Z3 solver fails to enumerate even a single model within 10,000 seconds. Therefore, we adopted the dynamic setting for all baselines to prevent solver enumeration artifacts from overwhelming the comparison. With this standardization, our evaluation isolates the effectiveness of each method's heuristic pruning rules and its end-to-end solving efficiency.
\subsubsection{Metrics}
On real-world microservice benchmarks, we assessed three metrics: 
\begin{itemize}[leftmargin=*,labelsep=0.5em]
    \item \textbf{Fault Injection Number:} the number of fault injections required to discover all valid combinatorial faults.
    \item \textbf{CNF-Solving Time:} the pure time spent solving CNF formulas to obtain all valid combinatorial faults.
    \item \textbf{End-to-end Time:} the full pipeline time, including workload generation, trace collection, CNF formula construction, solution computation, injection execution, failure validation, and system recovery.
\end{itemize}
On the large-scale simulated microservice benchmark, we focused on computational efficiency, measured by the CNF-solving time to discover all valid combinatorial faults for the system, as well as the fault injection number metric.

\subsubsection{Results}
\paragraph{(i) Results on Real-World Microservice Benchmarks.}

\begin{table*}[t]
\captionsetup{justification=centering}
\caption{Comparison of fault injection algorithms on real-world microservice benchmarks. \textit{Note:} OB-$n$/HR-$n$/SS-$n$/TT-$n$ denote Online Boutique/Hotel Reservation/Sock Shop/Train Ticket deployed with $n$ replicas per microservice.}
\label{tab:rq1-result}
\centering
\scriptsize
\begingroup
\setlength{\tabcolsep}{2pt}
\renewcommand{\arraystretch}{1.3} 
\large
\resizebox{\linewidth}{!}{
\begin{tabular}{c|l|cccc|cccc|cccc}
\toprule
\multirow{2}{*}{Benchmark} & \multirow{2}{*}{Request}
 & \multicolumn{4}{c|}{Fault Injection Number}
 & \multicolumn{4}{c|}{CNF-Solving Time (s)}
 & \multicolumn{4}{c}{End-to-end Time (s)} \\
\cline{3-6}\cline{7-10}\cline{11-14}
\rule{0pt}{3ex}
 & & LDFI & IntelliFI & MicroFI & FastFI
   & LDFI & IntelliFI & MicroFI & FastFI
   & LDFI & IntelliFI & MicroFI & FastFI \\
\midrule
\multirow{4}{*}{OB-4} & GET / &
        97 & 88 & 81 & \textbf{40} &
        5.20 & 4.62 & 4.44 & \textbf{1.05} &
        2299.08 & 2146.19 & 1929.09 & \textbf{985.14} \\
    & GET /product/\{id\} &
        101 & 94 & 84 & \textbf{41} &
        8.43 & 7.66 & 5.77 & \textbf{1.28} &
        2400.92 & 2178.54 & 1961.15 & \textbf{979.93} \\
    & GET /cart &
        121 & 116 & 100 & \textbf{46} &
        11.20 & 10.79 & 10.88 & \textbf{1.47} &
        2758.01 & 2658.94 & 2312.56 & \textbf{1075.65} \\
    & POST /checkout &
        177 & 136 & 152 & \textbf{53} &
        118.30 & 61.38 & 71.06 & \textbf{3.34} &
        4585.24 & 3507.73 & 3892.91 & \textbf{1416.17} \\
\midrule
\multirow{4}{*}{OB-6} & GET / &
        235 & 223 & 216 & \textbf{67} &
        46.65 & 36.31 & 35.08 & \textbf{7.20} &
        5380.31 & 5094.47 & 4923.50 & \textbf{1720.26} \\
    & GET /product/\{id\} &
        262 & 245 & 224 & \textbf{76} &
        123.87 & 96.32 & 75.53 & \textbf{10.82} &
        6378.36 & 6034.62 & 5498.09 & \textbf{1840.86} \\
    & GET /cart &
        281 & 247 & 228 & \textbf{78} &
        131.81 & 103.11 & 85.30 & \textbf{12.16} &
        6948.52 & 6135.43 & 5626.99 & \textbf{1933.31} \\
    & POST /checkout &
        413 & 266 & 288 & \textbf{106} &
        12424.17 & 9102.14 & 12002.79 & \textbf{63.98} &
        24747.15 & 16792.60 & 19224.28 & \textbf{2868.78} \\
\midrule
\multirow{3}{*}{HR-4} & GET /recommendations &
        26 & 13 & 17 & \textbf{10} &
        0.31 & 0.12 & 0.18 & \textbf{0.11} &
        574.27 & 325.88 & 441.26 & \textbf{231.67} \\
    & GET /hotels &
        122 & 75 & 97 & \textbf{23} &
        7.05 & 3.92 & 5.04 & \textbf{0.42} &
        2741.95 & 1747.87 & 2191.59 & \textbf{543.42} \\
    & POST /reservation &
        15 & 14 & 15 & \textbf{9} &
        0.18 & 0.17 & 0.17 & \textbf{0.14} &
        336.47 & 334.05 & 335.42 & \textbf{210.61} \\
\midrule
\multirow{3}{*}{HR-6} & GET /recommendations &
        60 & 34 & 50 & \textbf{16} &
        1.28 & 0.67 & 0.97 & \textbf{0.17} &
        1316.11 & 780.64 & 1101.60 & \textbf{370.94} \\
    & GET /hotels &
        236 & 202 & 222 & \textbf{31} &
        24.14 & 19.62 & 21.97 & \textbf{1.02} &
        7941.26 & 4844.69 & 5194.88 & \textbf{832.75} \\
    & POST /reservation &
        38 & 29 & 33 & \textbf{15} &
        0.77 & 0.41 & 0.59 & \textbf{0.16} &
        836.11 & 661.92 & 736.16 & \textbf{352.72} \\
\midrule
\multirow{3}{*}{SS-4} & GET /login &
        18 & 15 & 15 & \textbf{10} &
        0.19 & 0.15 & 0.15 & \textbf{0.09} &
        421.17 & 363.69 & 374.82 & \textbf{268.68} \\
    & POST /orders &
        301 & 141 & 173 & \textbf{35} &
        58.87 & 21.66 & 28.42 & \textbf{1.04} &
        7152.15 & 3349.78 & 4241.73 & \textbf{889.25} \\
    & POST /cart &
        21 & 12 & 18 & \textbf{10} &
        0.20 & 0.10 & 0.15 & \textbf{0.09} &
        462.88 & 285.75 & 377.40 & \textbf{210.48} \\
\midrule
\multirow{3}{*}{SS-6} & GET /login &
        66 & 27 & 30 & \textbf{15} &
        2.51 & 0.36 & 0.42 & \textbf{0.11} &
        1482.32 & 626.80 & 679.94 & \textbf{351.44} \\
    & POST /orders &
        463 & 280 & 324 & \textbf{52} &
        2990.70 & 2021.42 & 2280.13 & \textbf{10.14} &
        13392.38 & 8757.31 & 11102.41 & \textbf{1334.60} \\
    & POST /cart &
        54 & 25 & 25 & \textbf{16} &
        0.87 & 0.31 & 0.27 & \textbf{0.11} &
        1249.32 & 607.25 & 585.24 & \textbf{373.58} \\
\midrule
\multirow{6}{*}{TT-2} & POST /left &
        211 & 186 & 195 & \textbf{68} &
        53.40 & 48.68 & 49.60 & \textbf{6.82} &
        5251.36 & 4887.12 & 4990.69 & \textbf{1688.09} \\
    & POST /cheapest &
        372 & 343 & 354 & \textbf{156} &
        158.66 & 126.83 & 144.38 & \textbf{27.17} &
        10067.23 & 8623.13 & 9722.06 & \textbf{4571.71} \\
    & POST /minStation &
        351 & 333 & 342 & \textbf{151} &
        148.46 & 116.77 & 129.03 & \textbf{26.53} &
        9425.73 & 8113.86 & 8617.50 & \textbf{4393.10} \\
    & POST /quickest &
        368 & 341 & 357 & \textbf{153} &
        154.21 & 122.21 & 147.77 & \textbf{26.70} &
        10289.49 & 8532.38 & 9923.20 & \textbf{4521.05} \\
    & GET /food &
        33 & 27 & 28 & \textbf{16} &
        1.35 & 1.03 & 1.20 & \textbf{0.22} &
        679.45 & 577.80 & 584.87 & \textbf{289.76} \\
    & POST /preserve &
        348 & 338 & 340 & \textbf{115} &
        140.64 & 119.82 & 122.48 & \textbf{20.13} &
        23569.01 & 22311.96 & 22576.42 & \textbf{6160.43} \\
\midrule
\multirow{6}{*}{TT-3} & POST /left &
        1591 & 1565 & 1581 & \textbf{249} &
        1617.56 & 1531.13 & 1587.99 & \textbf{114.44} &
        38998.83 & 38251.02 & 38688.04 & \textbf{6613.93} \\
    & POST /cheapest &
        4013 & 3816 & 3680 & \textbf{1021} &
        18150.08 & 15224.18 & 15048.12 & \textbf{1165.69} &
        118985.63 & 111810.83 & 108374.78 & \textbf{27979.78} \\
    & POST /minStation &
        3660 & 3327 & 3337 & \textbf{917} &
        16492.12 & 12574.39 & 12732.02 & \textbf{1043.53} &
        108652.69 & 96597.06 & 97250.21 & \textbf{24791.49} \\
    & POST /quickest &
        4058 & 3902 & 3919 & \textbf{1087} &
        17812.31 & 15543.74 & 16932.20 & \textbf{1282.62} &
        119883.38 & 114479.03 & 116089.91 & \textbf{29283.32} \\
    & GET /food &
        43 & 32 & 38 & \textbf{22} &
        1.57 & 1.09 & 1.30 & \textbf{0.40} &
        1112.43 & 942.45 & 1021.76 & \textbf{660.15} \\
    & POST /preserve &
        3007 & 2863 & 2912 & \textbf{535} &
        15636.59 & 14876.04 & 15302.74 & \textbf{568.96} &
        195496.38 & 183225.47 & 189570.76 & \textbf{29215.41} \\
\bottomrule
\end{tabular}
}
\endgroup
\end{table*}

As shown in Table~\ref{tab:rq1-result}, FastFI achieves the best performance in terms of fault injection number, CNF-solving time, and end-to-end time.

\textbf{For the fault injection number}, MicroFI, IntelliFI, and LDFI require 608, 604, and 661 fault injections on average, respectively, whereas FastFI requires only 163. This corresponds to average reductions of 73.10\%, 72.93\%, and 75.24\%, respectively. This improvement is observed for the vast majority of requests in our benchmarks, although the magnitude varies by workload. The benefit is particularly pronounced for workloads with larger hypothesis spaces (e.g., OB-6 \texttt{POST /checkout}, HR-6 \texttt{GET /hotels}, SS-6 \texttt{POST /orders}, and most requests in TT-3), where the baseline heuristics still admit a large number of candidate combinations.
One key reason is that FastFI effectively manages candidate solutions across iterations. Unlike prior approaches that carry forward solutions derived from earlier CNF instances, FastFI treats each newly constructed CNF formula as a fresh search problem and invalidates previously enumerated candidates because their satisfiability can change as the CNF formula is updated by a newly observed execution path. Retaining such stale candidates can trigger repeated validations of hypotheses that are no longer satisfiable under the updated constraints, inflating the injection budget. By focusing solutions computed from the latest CNF instance and explicitly avoiding validation of outdated candidates, FastFI maintains a tighter candidate set and converges with substantially fewer injections while still discovering all valid combinatorial faults.

\textbf{For the CNF-solving time}, when a request triggers only a small set of APIs and yields short CNF clauses (e.g., SS-4 \texttt{GET /login}), FastFI’s per-instance CNF-solving time is comparable to the Z3-based baselines. As the number of APIs invoked per request increases, FastFI’s advantage widens, yielding progressively larger speedups over the Z3-based baselines. For instance, TT-3 \texttt{POST /preserve} invokes the most distinct APIs in our evaluation, and FastFI solves the resulting CNF instance in 568.96\,s, whereas the fastest baseline, IntelliFI, still requires 14,876.04\,s. 
Additionally, a complementary stressor is the number of CNF variables, which can grow even when the per-request API fan-out is not the largest. For example, OB-6 \texttt{POST /checkout} invokes 12 APIs, but the six replicas expand the variable space to 72 API variables. In this instance, FastFI spends 63.98\,s, while the fastest baseline IntelliFI needs 9,102.14\,s, i.e., 143$\times$ slower than FastFI.
On average, MicroFI, IntelliFI, and LDFI spend 2,400.88\,s, 2,243.04\,s, and 2,697.61\,s, respectively, in pure CNF-solving time, whereas FastFI requires only 137.44\,s, yielding reductions of 94.28\%, 93.87\%, and 94.91\%. These results indicate that FastFI’s search-space pruning reduces invalid injection trials while preserving fault coverage. In addition, FastFI’s DFS-based solver is markedly faster than the Z3-based solvers employed by MicroFI, IntelliFI, and LDFI.

\textbf{For the end-to-end time}, MicroFI, IntelliFI, and LDFI require on average 21,254.41\,s, 20,799.57\,s, and 22,994.24\,s, respectively, whereas FastFI completes in 4,967.45\,s, reducing end-to-end time by 76.63\%, 76.12\%, and 78.40\% compared with these baselines. These improvements are driven by two complementary effects: (i) FastFI validates substantially fewer fault hypotheses, reducing the cumulative cost of executing injections and collecting traces; and (ii) FastFI solves each CNF formula much faster with our DFS-based solver, avoiding the expensive Z3-based model enumeration that dominates the baselines when the candidate space is large. Consequently, FastFI consistently achieves the shortest end-to-end time across benchmarks, particularly for requests that invoke many APIs or contain large variable spaces.


\begin{figure}[t]
\captionsetup{skip=2pt}
  \centering
  \includegraphics[width=0.91\linewidth]{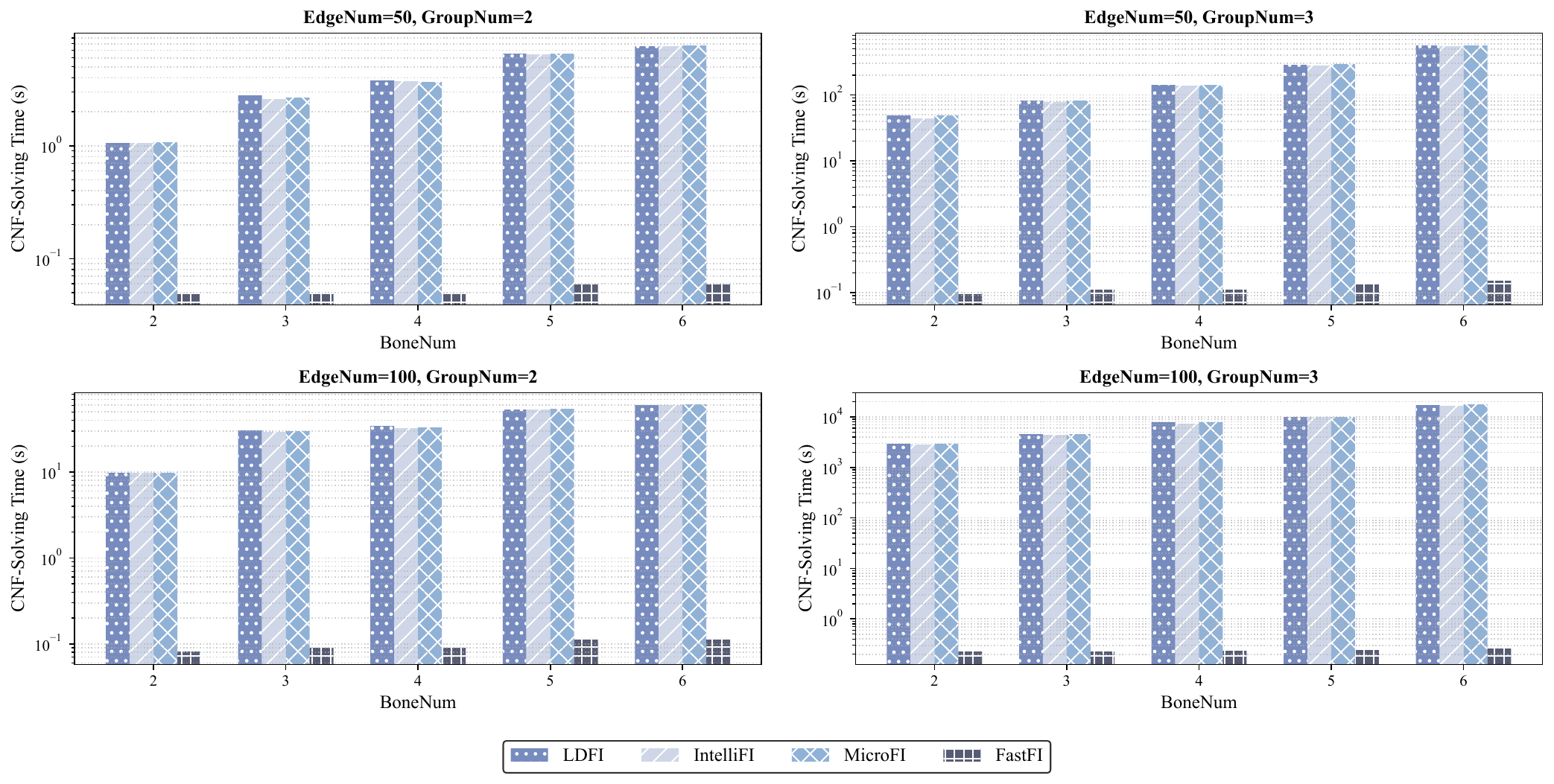}
  \caption{CNF-solving time for discovering all valid combinatorial faults under different API-dependency graph configurations.}
  \label{fig: RQ1_time_cost}
\end{figure}
\begin{figure}[t]
\captionsetup{skip=2pt}
  \centering
  \includegraphics[width=0.91\linewidth]{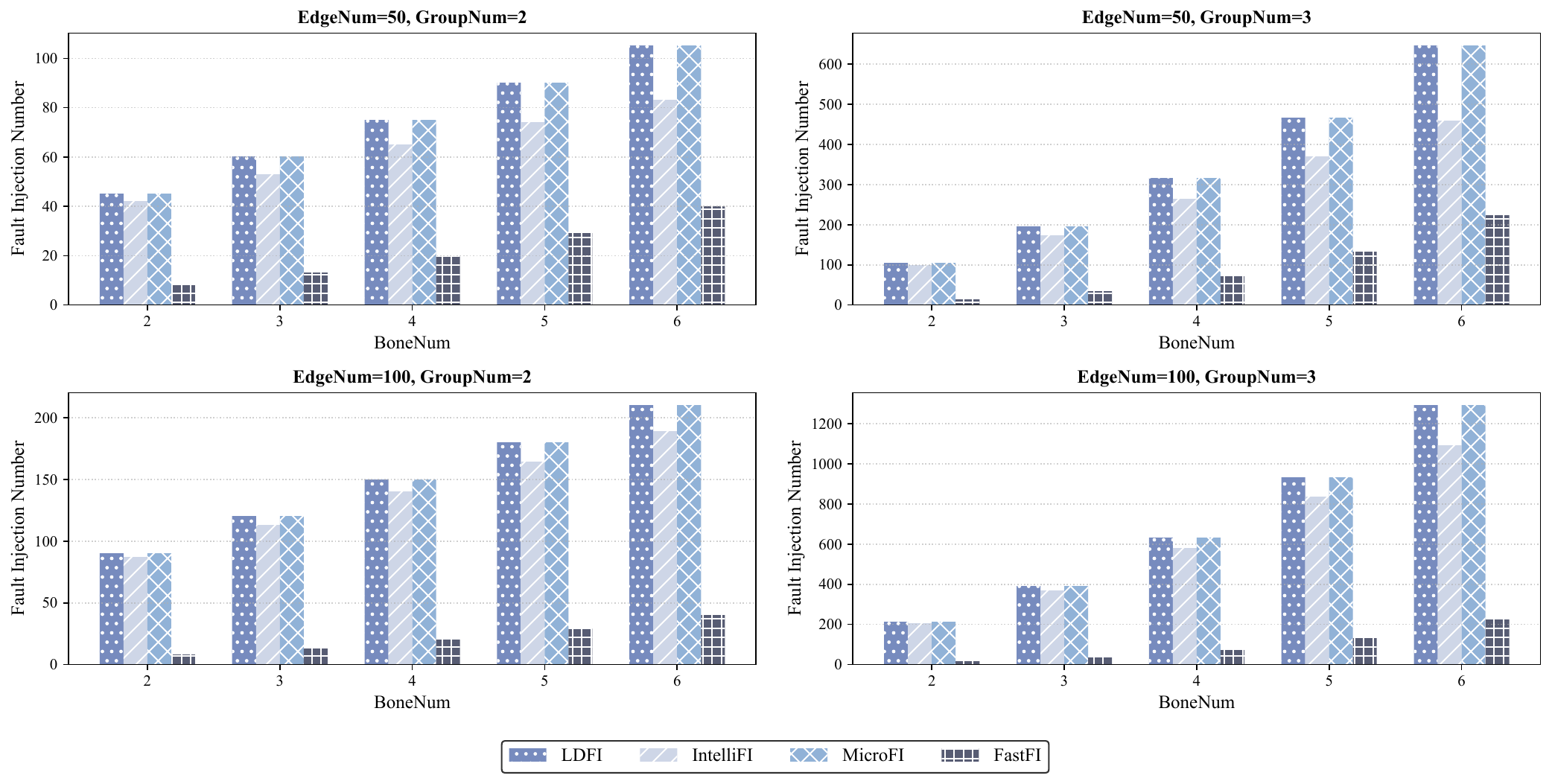}
  \caption{Fault injection number for discovering all valid combinatorial faults under different API-dependency graph configurations.}
  \label{fig: RQ1_fault_injection_num}
\end{figure}
\paragraph{(ii) Results on the Simulated Microservice Benchmark.}
As shown in Fig.~\ref{fig: RQ1_time_cost} and Fig.~\ref{fig: RQ1_fault_injection_num}, increasing the structural complexity of the per-request API-dependency graph by raising \texttt{EdgeNum} from 50 to 100, \texttt{GroupNum} from 2 to 3, and \texttt{BoneNum} from 2 to 6, rapidly enlarges the hypothesis space of combinatorial faults. Intuitively, denser dependencies and richer grouping/bone structures introduce more interacting API subsets and constraints, which in turn yield larger (and harder) CNF instances and a larger set of candidate faults to validate. Under this scaling, the Z3-based baselines exhibit a sharp blow-up in both fault injection number and CNF-solving time, indicating that their exploration cost grows rapidly with problem complexity.
For \texttt{EdgeNum}=50, \texttt{GroupNum}=2, and \texttt{BoneNum}=2, the baselines already require 45--60 injections and 1.05--1.06\,s, whereas FastFI needs only 8 injections and 0.05\,s. When scaling to \texttt{EdgeNum}=100, \texttt{GroupNum}=3, and \texttt{BoneNum}=6, the baselines increase to 1,092--1,290 injections and 16,571.27--17,193.21\,s, while FastFI requires 222 injections and 0.26\,s, achieving the same fault coverage with a substantially smaller injection budget. Notably, the gap widens as complexity increases, suggesting that FastFI mitigates the combinatorial explosion more effectively than these baselines.
This advantage stems from FastFI’s aggressive pruning together with the efficiency of its DFS-based solver. In contrast, MicroFI, IntelliFI, and LDFI continue to validate candidates that were SAT solutions to earlier CNF formulas even after the formula is updated; as a result, they spend additional injections and solver effort on candidates that have become invalid under the updated formula, amplifying both runtime and injection cost at larger scales.

\begin{tcolorbox}[
  enhanced,
  colback=gray!5,
  colframe=gray!500,
  arc=2mm,
  boxrule=1.5pt,
  left=1.5mm,right=1.5mm,top=1mm,bottom=1mm
]
\textbf{Answer to RQ1: }
FastFI is substantially more efficient at identifying valid fault combinations. On real-world benchmarks, it reduces fault injection number trials by 72.93\%--75.24\% on average and shortens CNF-solving time by 93.87\%--94.91\%, yielding a 76.12\%--78.40\% reduction in end-to-end time. These gains stem from aggressively discarding stale candidates after each CNF formula update and leveraging a DFS-based solver that efficiently enumerates minimal satisfying assignments while avoiding expensive Z3-based model enumeration.
\end{tcolorbox}

\subsection{RQ2: Performance of the DFS-Based Solver for Monotone CNF Formulas}
\subsubsection{Baselines}
To evaluate the performance of FastFI's DFS-based solver on monotone CNF formulas, we compare it against three representative SAT solvers:
\begin{itemize}[leftmargin=*,labelsep=0.5em]
    \item \textbf{Z3} ~\cite{z3_paper}. It is a widely used SAT solver developed by Microsoft Research. It has also been adopted as the SAT solver in prior fault injection methods, including LDFI, IntelliFI, and MicroFI. When solving SAT instances, Z3 employs a Davis--Putnam--Logemann--Loveland (DPLL) search procedure augmented with conflict-driven clause learning (CDCL)~\cite{DPLL,CDCL}, featuring conflict analysis, clause learning, and non-chronological backtracking. As a general-purpose solver, Z3 is engineered to support a broad range of constraints beyond propositional satisfiability, which makes it a strong reference point for assessing how a general solver performs on our specialized monotone CNF formula.
    \item \textbf{AE-Kissat-MAB} ~\cite{AE-kissat-MAB}. It ranked first in the Main Track of the SAT Competition 2025. It builds on Kissat-public and introduces an LLM-driven solver-evolution pipeline, where large language models guide the discovery and refinement of configurations, combined with multi-armed bandit (MAB) strategies to efficiently explore. This baseline represents an emerging paradigm that leverages LLMs to systematically enhance CDCL solvers.
    \item \textbf{Kissat-public}~\cite{Kissat,kissat_paper}.  It ranked second in the Main Track of the SAT Competition 2025. and is widely regarded as a high-performance, state-of-the-art CDCL solver. It incorporates aggressive in-processing (e.g., bounded variable elimination, subsumption, and vivification), carefully tuned restart policies, and optimized data structures for efficient unit propagation. Kissat-public has also served as the foundation for several competitive SAT solvers.
\end{itemize}

Together, these baselines enable a well-rounded comparison with FastFI’s specialized DFS-based CNF-solving approach. To isolate the solver’s effect, we kept all other FastFI components fixed (including pruning and dynamic mechanisms) and replaced only the SAT solver during evaluation.

\subsubsection{Metrics} On both the real-world microservice benchmarks and the large-scale simulated microservice benchmark, we reported \textbf{CNF-solving time} and the \textbf{Average Clause Overlap} ($ACO$) metric, which quantifies the overlap among alternative execution paths for each request.
For each request, we constructed a monotone CNF formula $F=\bigwedge_{i=1}^{m} C_i$, where each clause $C_i\subseteq X$ corresponds to an alternative execution path and $X$ is the set of all Boolean variables.
To compute $ACO$, we first calculate the average clause length as $\bar{\ell}=\frac{1}{m}\sum_{i=1}^{m}|C_i|$.
For each variable $x\in X$, we define its clause coverage $\deg(x)$ as the number of clauses that contain $x$:
$\deg(x)=|\{\, i\in\{1,\dots,m\}\mid x\in C_i \,\}|$.
We then define $ACO$ as the normalized average intersection size between two distinct clauses:
\begin{equation}
ACO=\frac{2}{m(m-1)\bar{\ell}}\sum_{x\in X}\binom{\deg(x)}{2}.
\end{equation}
Intuitively, $ACO$ measures the average fraction of variables shared by two uniformly randomly selected alternative paths. A larger $ACO$ indicates heavier overlap among alternative execution paths, whereas a smaller $ACO$ suggests that the paths are closer to being disjoint.

\subsubsection{Results}
\paragraph{(i) Results on Real-World Microservice Benchmarks.}
\begin{table*}[t]
\centering
\captionsetup{justification=centering}
\caption{CNF-solving time comparison across solvers on real-world microservice benchmarks. \textit{Note:} OB-$n$/HR-$n$/SS-$n$/TT-$n$ denote Online Boutique/Hotel Reservation/Sock Shop/Train Ticket deployed with $n$ replicas per microservice.}
\label{tab:rq2_real_benchmark_solving_time}
\footnotesize
\begingroup
\setlength{\tabcolsep}{14pt}
\renewcommand{\arraystretch}{1.3}
{%
\begin{tabular}{c|l|c|cccc}
\toprule
\multirow{2}{*}[-1.5ex]{Benchmark} & \multirow{2}{*}[-1.5ex]{Request} & \multirow{2}{*}[-1.5ex]{$ACO$} &
\multicolumn{4}{c}{CNF-Solving Time (s)} \\
\cline{4-7}
& & & \rule{0pt}{4.5ex}Z3
  & \makecell[c]{AE-Kissat-\\MAB}
  & \makecell[c]{Kissat-\\public}
  & \makecell[c]{DFS-\\based} \\
\midrule
\multirow{4}{*}{OB-4} & GET /                 & 0.326608 & 1.48   & 17.94   & 20.83   & \textbf{1.05}  \\
                     & GET /product/\{id\}    & 0.315936 & 2.97   & 39.42   & 40.15   & \textbf{1.28}  \\
                     & GET /cart              & 0.314231 & 4.26   & 61.13   & 65.27   & \textbf{1.47}  \\
                     & POST /checkout         & 0.302119 & 23.63  & 224.79  & 232.47  & \textbf{3.34}  \\
\midrule
\multirow{4}{*}{OB-6} & GET /                 & 0.308554 & 8.58   & 85.07   & 94.03   & \textbf{7.20}  \\
                     & GET /product/\{id\}    & 0.269176 & 20.31  & 349.08  & 420.33  & \textbf{10.82} \\
                     & GET /cart              & 0.261836 & 26.38  & 359.65  & 455.24  & \textbf{12.16} \\
                     & POST /checkout         & 0.222716 & 959.91 & 1884.22 & 2423.28 & \textbf{63.98} \\
\midrule
\multirow{3}{*}{HR-4} & GET /recommendations  & 0.142857 & 0.15   & 0.99    & 1.29    & \textbf{0.11}  \\
                     & GET /hotels           & 0.219883 & 0.63   & 18.68   & 28.14   & \textbf{0.42}  \\
                     & POST /reservation     & 0.166667 & 0.17   & 0.82    & 1.02    & \textbf{0.14}  \\
\midrule
\multirow{3}{*}{HR-6} & GET /recommendations  & 0.128571 & 0.22   & 3.11    & 3.48    & \textbf{0.17}  \\
                     & GET /hotels           & 0.154985 & 3.36   & 86.78   & 99.32   & \textbf{1.02}  \\
                     & POST /reservation     & 0.138095 & 0.19   & 2.41    & 3.09    & \textbf{0.16}  \\
\midrule
\multirow{3}{*}{SS-4} & GET /login            & 0.190909 & 0.11   & 0.73    & 0.84    & \textbf{0.09}  \\
                     & POST /orders          & 0.254938 & 4.83   & 115.68  & 191.56  & \textbf{1.04}  \\
                     & POST /cart            & 0.166667 & 0.10   & 0.77    & 0.94    & \textbf{0.09}  \\
\midrule
\multirow{3}{*}{SS-6} & GET /login            & 0.147061 & 0.14   & 3.02    & 3.73    & \textbf{0.11}  \\
                     & POST /orders          & 0.203158 & 100.36 & 768.89  & 1507.23 & \textbf{10.14} \\
                     & POST /cart            & 0.128576 & 0.13   & 3.17    & 3.91    & \textbf{0.11}  \\
\midrule
\multirow{6}{*}{TT-2} & POST /left            & 0.728957 & 15.37  & 196.17  & 222.96  & \textbf{6.82}  \\
                     & POST /cheapest        & 0.829188 & 40.63  & 684.42  & 799.12  & \textbf{27.17} \\
                     & POST /minStation      & 0.792552 & 37.95  & 662.05  & 716.65  & \textbf{26.53} \\
                     & POST /quickest        & 0.839739 & 39.02  & 687.25  & 829.10  & \textbf{26.70} \\
                     & POST /food            & 0.508571 & 0.26   & 2.54    & 2.68    & \textbf{0.22}  \\
                     & POST /preserve        & 0.766868 & 37.07  & 545.28  & 684.75  & \textbf{20.13} \\
\midrule
\multirow{6}{*}{TT-3} & POST /left            & 0.663501 & 206.22  & 5441.03  & 6984.51  & \textbf{114.44}  \\
                     & POST /cheapest        & 0.809494 & 3121.09 & 80612.69 & 85418.04 & \textbf{1165.69} \\
                     & POST /minStation      & 0.799367 & 2843.54 & 79215.39 & 80938.54 & \textbf{1043.53} \\
                     & POST /quickest        & 0.816148 & 3031.39 & 81739.31 & 85795.87 & \textbf{1282.62} \\
                     & POST /food            & 0.349902 & 0.54    & 8.18     & 8.73     & \textbf{0.40}    \\
                     & POST /preserve        & 0.711351 & 1556.73 & 43401.44 & 50696.11 & \textbf{568.96}  \\
\bottomrule
\end{tabular}%
}
\endgroup
\end{table*}

Table~\ref{tab:rq2_real_benchmark_solving_time} shows the time spent in solving CNF formulas by different solvers on real-world microservice benchmarks. The DFS-based solver consistently outperforms all baselines, and it is the fastest on every instance. On average, the DFS-based solver finishes in 137.44\,s, whereas Z3, AE-Kissat-MAB, and Kissat-public take 377.74\,s, 9,288.19\,s, and 9,959.16\,s, respectively, corresponding to being 2.75$\times$, 67.58$\times$, and 72.46$\times$ slower than the DFS-based solver, respectively.
These baseline solvers follow the workflow below: to obtain all satisfying assignments, they repeatedly invoke the solver and add a blocking clause to exclude each previously found model. While each solver invocation can be fast, some real-world microservice benchmarks yield a large number of candidates; thus, the cumulative overhead of thousands of solver invocations on the increasingly large CNF formula dominates the total execution time. In contrast, the DFS-based solver performs a single DFS traversal with batch evaluation and leverages bitmask operations, along with pruning, to avoid redundant exploration. As a result, it substantially reduces the CNF-solving time, especially on the hardest instances where AE-Kissat-MAB and Kissat-public can take 81,739--85,796\,s ($\approx$22.7--23.8\,h), while the DFS-based solver finishes in 1,282.62\,s ($\approx$21.4\, min).

Regarding the $ACO$ metric, Table~\ref{tab:rq2_real_benchmark_solving_time} reports that the real-world microservice benchmarks exhibit relatively large $ACO$ values (0.41 on average). As described in Section~\ref{sec: benchmarks}, following MicroFI and IntelliFI, we constructed alternative execution paths via replica-based instance-level failover to ensure comparability. This design naturally increases shared dependencies across alternatives, leading to heavier clause overlap in the resulting CNF formulas. Moreover, failover across replicas introduces additional options, thereby increasing the number of clauses in the CNF formula. As a result, the CNF formulas derived from real-world microservice benchmarks are characterized by both a large clause count and a high degree of clause overlap. According to our theoretical analysis (Theorem~\ref{proof:time_complex}), under such a setting, the clause count becomes the dominant factor that affects the overall solving cost, which in turn increases the solving cost. Nevertheless, even in this unfavorable regime with many clauses and heavy overlap, our DFS-based solver consistently outperforms the other solvers.

\paragraph{(ii) Results on the Simulated Microservice Benchmark.}
\begin{table*}[t]
\centering
\caption{$ACO$ results under different API-dependency graph configurations for each URL request.}
\label{tab:rq2_sim_aco}
\footnotesize
\begingroup
\setlength{\tabcolsep}{4pt}
\renewcommand{\arraystretch}{1.1}
\begin{tabular}{c|ccc|c|ccc|c|ccc}
\toprule
\multirow{2}{*}[0ex]{EdgeNum} & \multicolumn{3}{c|}{Configuration} &
\multirow{2}{*}[0ex]{EdgeNum} & \multicolumn{3}{c|}{Configuration} &
\multirow{2}{*}[0ex]{EdgeNum} & \multicolumn{3}{c}{Configuration} \\
\cline{2-4}\cline{6-8}\cline{10-12}
\rule{0pt}{3ex}
& GroupNum & BoneNum & $ACO$ & & GroupNum & BoneNum & $ACO$ & & GroupNum & BoneNum & $ACO$ \\
\midrule
\multirow{6}{*}[-1.5ex]{50}  & \multirow{3}{*}[0ex]{2} & 2 & 0.0267 &
\multirow{6}{*}[-1.5ex]{100} & \multirow{3}{*}[0ex]{2} & 2 & 0.0133 &
\multirow{6}{*}[-1.5ex]{150} & \multirow{3}{*}[0ex]{2} & 2 & 0.0089 \\
                            &                             & 3 & 0.0400 &
                            &                             & 3 & 0.0200 &
                            &                             & 3 & 0.0133 \\
                            &                             & 4 & 0.0533 &
                            &                             & 4 & 0.0267 &
                            &                             & 4 & 0.0178 \\
\cmidrule(l{2pt}r{2pt}){2-4}\cmidrule(l{2pt}r{2pt}){6-8}\cmidrule(l{2pt}r{2pt}){10-12}
                            & \multirow{3}{*}[0ex]{3} & 2 & 0.0160 &
                            & \multirow{3}{*}[0ex]{3} & 2 & 0.0080 &
                            & \multirow{3}{*}[0ex]{3} & 2 & 0.0053 \\
                            &                            & 3 & 0.0240 &
                            &                            & 3 & 0.0120 &
                            &                            & 3 & 0.0080 \\
                            &                            & 4 & 0.0320 &
                            &                            & 4 & 0.0160 &
                            &                            & 4 & 0.0107 \\
\midrule
\multirow{6}{*}[-1.5ex]{200} & \multirow{3}{*}[0ex]{2} & 2 & 0.0067 &
\multirow{6}{*}[-1.5ex]{250} & \multirow{3}{*}[0ex]{2} & 2 & 0.0053 &
\multirow{6}{*}[-1.5ex]{300} & \multirow{3}{*}[0ex]{2} & 2 & 0.0044 \\
                            &                             & 3 & 0.0100 &
                            &                             & 3 & 0.0080 &
                            &                             & 3 & 0.0067 \\
                            &                             & 4 & 0.0133 &
                            &                             & 4 & 0.0107 &
                            &                             & 4 & 0.0089 \\
\cmidrule(l{2pt}r{2pt}){2-4}\cmidrule(l{2pt}r{2pt}){6-8}\cmidrule(l{2pt}r{2pt}){10-12}
                            & \multirow{3}{*}[0ex]{3} & 2 & 0.0040 &
                            & \multirow{3}{*}[0ex]{3} & 2 & 0.0032 &
                            & \multirow{3}{*}[0ex]{3} & 2 & 0.0027 \\
                            &                            & 3 & 0.0060 &
                            &                            & 3 & 0.0048 &
                            &                            & 3 & 0.0040 \\
                            &                            & 4 & 0.0080 &
                            &                            & 4 & 0.0064 &
                            &                            & 4 & 0.0053 \\
\bottomrule
\end{tabular}
\endgroup
\end{table*}

Table~\ref{tab:rq2_sim_aco} shows the $ACO$ results on the simulated microservice benchmark. This benchmark is designed to more closely approximate production-grade high-reliability practices. To improve the availability of alternative execution paths, system designers often aim to reduce common-mode failures by limiting shared dependencies across alternatives~\cite{low_overlap1, low_overlap2, low_overlap3}, so that a failure on one path is less likely to invalidate the others. Consequently, the overlap among alternative paths is expected to be relatively low by design. Consistent with this intuition, the simulated benchmark yields significantly smaller $ACO$ values than the real-world benchmarks, whose alternative execution paths are constructed via replica-based failover. This indicates weaker clause overlap in the resulting CNF formulas. This observation aligns with the common production practice of decoupling alternatives to enhance availability, and it also suggests that the simulated benchmark better captures the low-overlap multipath structure in production environments.

\begin{figure}[t]
\captionsetup{skip=2pt}
  \centering
  \includegraphics[width=\linewidth]{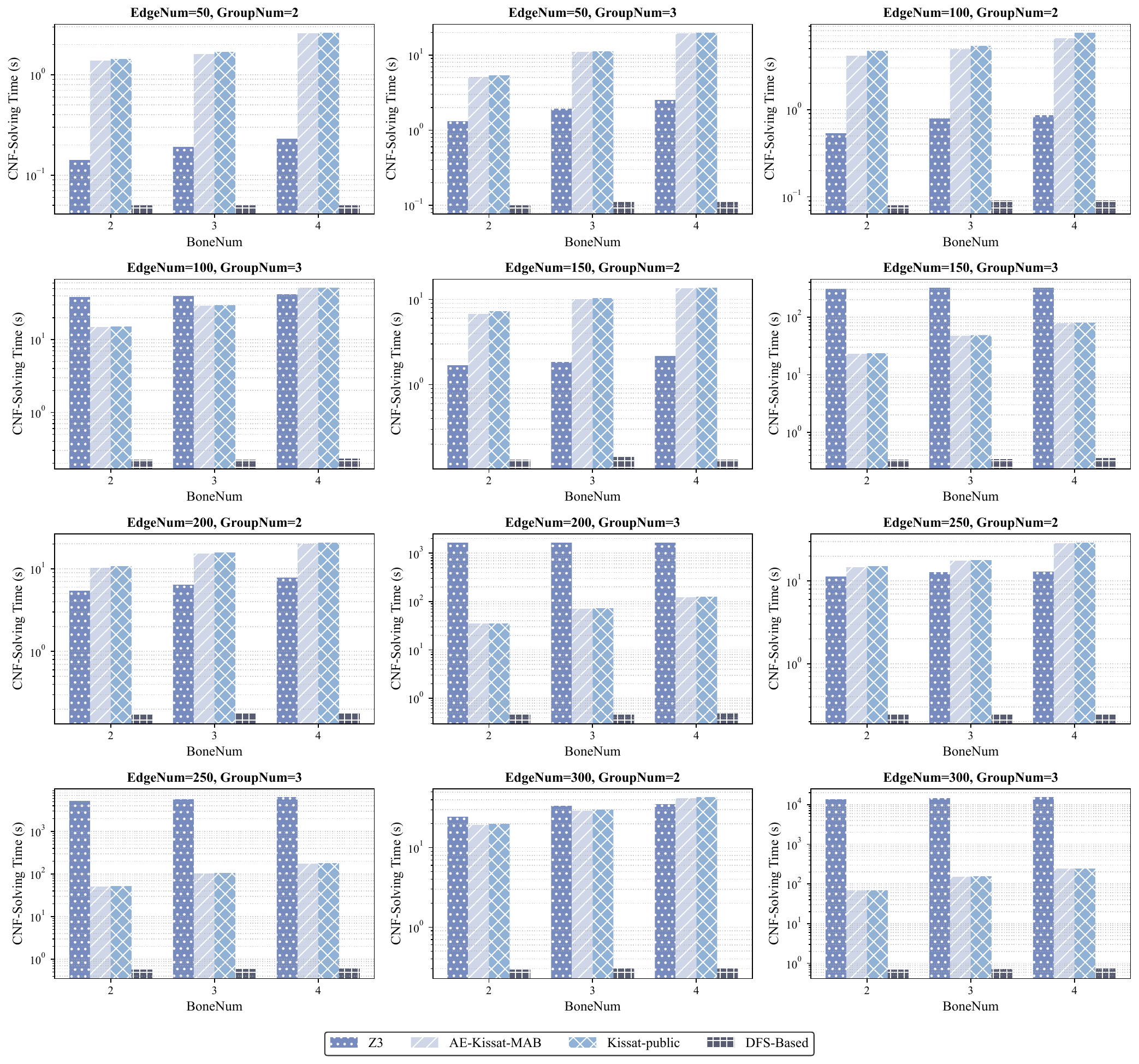}
  \caption{CNF-solving time across solvers under different API-dependency graph configurations.}
  \label{fig: RQ2_time_cost}
\end{figure}
Fig.~\ref{fig: RQ2_time_cost} shows the CNF-solving time on the simulated microservice benchmark across 36 configurations. The DFS-based solver consistently achieves the lowest CNF-solving time, remaining below 1\,s across all configurations, ranging from 0.05\,s to 0.73\,s, with an average of 0.28\,s. In contrast, the baseline solvers are orders of magnitude slower, with average solving times of 1,834.44\,s for Z3, 42.37\,s for AE-Kissat-MAB, and 43.45\,s for Kissat-public. Overall, the DFS-based solver delivers 6,551.57$\times$, 151.32$\times$, and 155.18$\times$ speedups compared to Z3, AE-Kissat-MAB, and Kissat-public, respectively.

As \texttt{EdgeNum} and \texttt{GroupNum} increase, the CNF formula becomes more complex, and all solvers slow down. However, the DFS-based solver degrades much more gracefully. In addition, we observed a clear trend under \texttt{GroupNum}=3 as \texttt{EdgeNum} increases from 100 to 300: AE-Kissat-MAB and Kissat-public progressively outperform Z3. This is because, in this configuration, the dominant cost is no longer driven by enumerating all satisfying assignments; instead, it is driven by per-instance solving hardness. Specifically, the increase in the number of variables and clauses substantially raises the structural complexity of the CNF instance, making per-call SAT solving efficiency the primary bottleneck. Since AE-Kissat-MAB and Kissat-public are considerably faster at solving individual SAT instances on complex CNF formulas than Z3, their advantage becomes more pronounced as the CNF formula becomes more complex. Nevertheless, the DFS-based solver remains robust and completes all configurations within one second, even in the most challenging configuration (\texttt{EdgeNum}=300, \texttt{GroupNum}=3, \texttt{BoneNum}=4). Specifically, the DFS-based solver finishes in 0.73\,s, while Z3 requires 15,049.41\,s, AE-Kissat-MAB takes 236.81\,s, and Kissat-public takes 243.00\,s. These results confirm that the DFS-based solver scales favorably with increasing instance size.

\begin{tcolorbox}[
  enhanced,
  colback=gray!5,
  colframe=gray!500,
  arc=2mm,
  boxrule=1.5pt,
  left=1.5mm,right=1.5mm,top=1mm,bottom=1mm
]
\textbf{Answer to RQ2: }
Across both real-world and simulated microservice benchmarks, our DFS-based solver achieves the lowest solving time. On real-world benchmarks, it averages 137.44\,s, providing 2.75$\times$, 67.57$\times$, and 72.46$\times$ speedups over Z3, AE-Kissat-MAB, and Kissat-public, respectively. On the simulated benchmark, it remains below 1\,s (0.28\,s on average) and achieves 6,551.57$\times$, 151.32$\times$, and 155.18$\times$ speedups over the same baselines. These stem from a single DFS traversal with batched bitmask operations and aggressive pruning, which avoids repeated solver invocations by incrementally excluding explored candidates during the DFS search.
\end{tcolorbox}

\subsection{RQ3: Effectiveness of Dynamic Injection Mechanisms}
\subsubsection{Baselines}
To evaluate the effectiveness of FastFI's dynamic injection mechanism, we constructed a baseline by keeping FastFI's solver and pruning strategy unchanged, and replacing only the dynamic injection mechanism with a static injection scheme. This baseline is denoted as FastFI w/o Dynamic. A key parameter in static injection is the maximum number of APIs included in each combinatorial fault. If the parameter is set too small, we may miss valid combinatorial faults; if the parameter is set too large, it can dramatically expand the search space and incur substantial solving overhead. To identify valid combinatorial faults while keeping the search space tractable, we set this parameter to the replica count plus one for each real-world microservice benchmark. For the simulated microservice benchmark, we set this parameter to match \texttt{GroupNum}. This choice avoids an unnecessarily inflated search space, yielding a representative and directly comparable static baseline.

\subsubsection{Metrics}
On real-world microservice benchmarks, we evaluated three metrics consistent with RQ1: \textbf{fault injection number}, \textbf{CNF-solving time}, and \textbf{end-to-end time}. On the simulated microservice benchmark, we reported the CNF-solving time to discover all valid combinatorial faults.

\subsubsection{Results}
\paragraph{(i) Results on Real-World Microservice Benchmarks.}
\begin{table*}[t]
\captionsetup{justification=centering}
\caption{Effect of dynamic fault injection mechanism for FastFI’s performance on real-world microservice benchmarks. \textit{Note:} OB-$n$/HR-$n$/SS-$n$/TT-$n$ denote Online Boutique/Hotel Reservation/Sock Shop/Train Ticket deployed with $n$ replicas per microservice.}
\label{tab:rq3-real-world-result}
\centering
\footnotesize
\begingroup
\setlength{\tabcolsep}{8pt}
\renewcommand{\arraystretch}{1.3}
{%
\begin{tabular}{c|l|cc|cc|cc}
\toprule
\multirow{2}{*}[-1.5ex]{Benchmark} & \multirow{2}{*}[-1.5ex]{Request}
 & \multicolumn{2}{c}{Fault Injection Number}
 & \multicolumn{2}{|c|}{CNF-Solving Time (s)}
 & \multicolumn{2}{c}{End-to-end Time (s)} \\
\cline{3-4}\cline{5-6}\cline{7-8}
\rule{0pt}{4.5ex}
 & & \makecell[c]{FastFI\\w/o Dynamic} & FastFI
   & \makecell[c]{FastFI\\w/o Dynamic} & FastFI
   & \makecell[c]{FastFI\\w/o Dynamic} & FastFI \\
\midrule
\multirow{4}{*}{OB-4} 
& GET /                     & 64  & \textbf{40}  & 2.78    & \textbf{1.05}  & 1596.15  & \textbf{985.14}   \\
& GET /product/\{id\}       & 79  & \textbf{41}  & 8.47    & \textbf{1.28}  & 1926.15  & \textbf{979.93}   \\
& GET /cart                 & 82  & \textbf{46}  & 9.45    & \textbf{1.47}  & 1957.76  & \textbf{1075.65}  \\
& POST /checkout            & 105 & \textbf{53}  & 133.79  & \textbf{3.34}  & 2976.73  & \textbf{1416.17}  \\
\midrule
\multirow{4}{*}{OB-6} 
& GET /                     & 111 & \textbf{67}  & 34.78   & \textbf{7.20}  & 2863.39  & \textbf{1720.26}  \\
& GET /product/\{id\}       & 121 & \textbf{76}  & 370.86  & \textbf{10.82} & 3320.92  & \textbf{1840.86}  \\
& GET /cart                 & 135 & \textbf{78}  & 506.13  & \textbf{12.16} & 3887.21  & \textbf{1933.31}  \\
& POST /checkout            & 159 & \textbf{106} & 5603.88 & \textbf{63.98} & 9845.72  & \textbf{2868.78}  \\
\midrule
\multirow{3}{*}{HR-4} 
& GET /recommendations      & 15 & \textbf{10} & 0.19 & \textbf{0.11} & 345.08 & \textbf{231.67} \\
& GET /hotels               & 38 & \textbf{23} & 1.13 & \textbf{0.42} & 914.71 & \textbf{543.42} \\
& POST /reservation         & 13 & \textbf{9}  & 0.16 & \textbf{0.14} & 305.99 & \textbf{210.61} \\
\midrule
\multirow{3}{*}{HR-6} 
& GET /recommendations      & 24 & \textbf{16} & 0.29 & \textbf{0.17} & 561.05 & \textbf{370.94} \\
& GET /hotels               & 54 & \textbf{31} & 9.52 & \textbf{1.02} & 1467.16 & \textbf{832.75} \\
& POST /reservation         & 22 & \textbf{15} & 0.24 & \textbf{0.16} & 522.03 & \textbf{352.72} \\
\midrule
\multirow{3}{*}{SS-4} 
& GET /login                & 15 & \textbf{10} & 0.11  & \textbf{0.09} & 402.79 & \textbf{268.68} \\
& POST /orders              & 76 & \textbf{35} & 98.99 & \textbf{1.04} & 2035.42 & \textbf{889.25} \\
& POST /cart                & 14 & \textbf{10} & 0.10  & \textbf{0.09} & 306.18 & \textbf{210.48} \\
\midrule
\multirow{3}{*}{SS-6} 
& GET /login               & 21 & \textbf{15} & 0.25   & \textbf{0.11}  & 503.69  & \textbf{351.44}  \\
& POST /orders              & 95 & \textbf{52} & 463.05 & \textbf{10.14} & 2964.11 & \textbf{1334.60} \\
& POST /cart                & 23 & \textbf{16} & 0.39   & \textbf{0.11}  & 523.31  & \textbf{373.58}  \\
\midrule
\multirow{6}{*}{TT-2} 
& POST /left                & 217 & \textbf{68}  & 57.07  & \textbf{6.82}  & 5809.79  & \textbf{1688.09}  \\
& POST /cheapest            & 476 & \textbf{156} & 352.03 & \textbf{27.17} & 14004.17 & \textbf{4571.71}  \\
& POST /minStation          & 452 & \textbf{151} & 330.47 & \textbf{26.53} & 13532.28 & \textbf{4393.10}  \\
& POST /quickest            & 461 & \textbf{153} & 338.71 & \textbf{26.70} & 13961.24 & \textbf{4521.05}  \\
& POST /food                & 34  & \textbf{16}  & 0.53   & \textbf{0.22}  & 619.62   & \textbf{289.76}   \\
& POST /preserve            & 484 & \textbf{115} & 524.09 & \textbf{20.13} & 27547.23 & \textbf{6160.43}  \\
\midrule
\multirow{6}{*}{TT-3} 
& POST /left                & 773  & \textbf{249}  & 3077.78  & \textbf{114.44}  & 22350.11  & \textbf{6613.93}  \\
& POST /cheapest            & 3849 & \textbf{1021} & 52105.45 & \textbf{1165.69} & 154206.19 & \textbf{27979.78} \\
& POST /minStation          & 2608 & \textbf{917}  & 40637.55 & \textbf{1043.53} & 109013.96 & \textbf{24791.49} \\
& POST /quickest            & 4020 & \textbf{1087} & 61084.06 & \textbf{1282.62} & 166152.43 & \textbf{29283.32} \\
& POST /food                & 41   & \textbf{22}   & 0.78     & \textbf{0.40}    & 1143.27   & \textbf{660.15}   \\
& POST /preserve            & 1894 & \textbf{535}  & 24121.56 & \textbf{568.96}  & 127661.04 & \textbf{29215.41} \\
\bottomrule
\end{tabular}%
}\endgroup
\end{table*}

As shown in Table~\ref{tab:rq3-real-world-result}, FastFI with the dynamic injection mechanism consistently outperforms FastFI w/o Dynamic across all three metrics. 
\textbf{For the fault injection number}, FastFI w/o Dynamic requires 517 injections on average, which is 3.17$\times$ higher than FastFI.
This is because the static scheme fixes a relatively large upper bound on the number of APIs included in each combinatorial fault from the beginning, which generates many more candidate solutions and thereby increases the number of injection attempts.
\textbf{For the CNF-solving time}, even with the DFS-based solver, FastFI w/o Dynamic still spends 5,933.58\,s on average for CNF-solving time, which is 43.17$\times$ slower than FastFI. This gap stems from how the two mechanisms set the upper bound on the number of APIs in a combinatorial fault. The dynamic mechanism starts from a small bound and gradually increases it based on feedback from injection outcomes, which tightly constrains the solver search space early on.
In contrast, the static scheme uses a large bound from the beginning, imposing weaker constraints and thus generating many more candidate assignments; enumerating the enlarged solution space substantially increases solving time.
\textbf{For the end-to-end time}, consistent with the trends in fault injection number and CNF-solving time, FastFI w/o Dynamic incurs a substantially higher end-to-end time. On average, it takes 21,725.84\,s, which is 4.37$\times$ that of FastFI, due to the combined overhead of more injection attempts and the increased CNF-solving time.

\paragraph{(ii) Results on the Simulated Microservice Benchmark.}
\begin{figure}[t]
\captionsetup{skip=2pt}
  \centering
  \includegraphics[width=\linewidth]{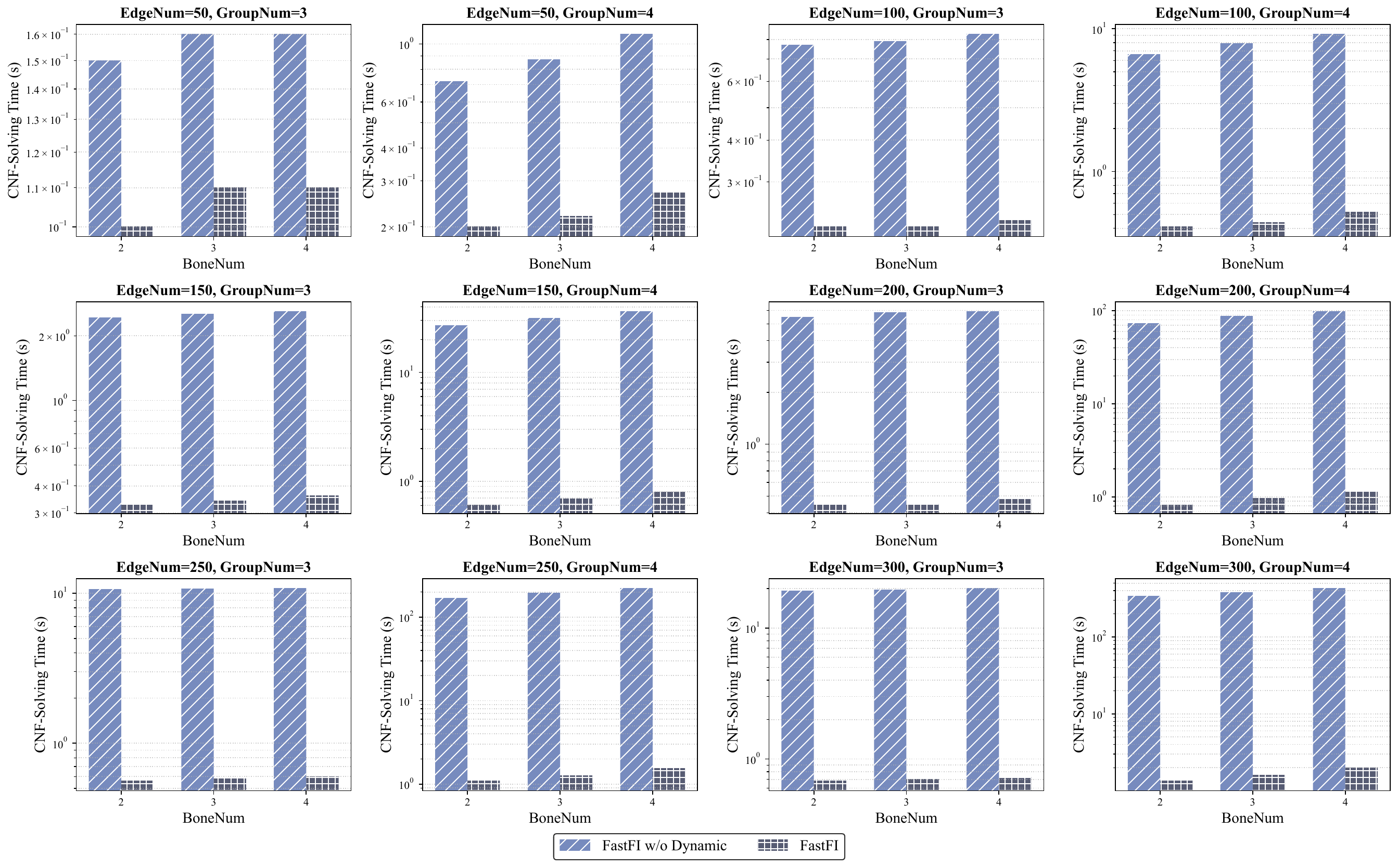}
  \caption{Impact of dynamic fault injection mechanism on CNF-solving time across different configurations.}
  \label{fig: RQ3_time_cost}
\end{figure}
Fig.~\ref{fig: RQ3_time_cost} compares the CNF-solving time of FastFI and FastFI w/o Dynamic on the simulated microservice benchmark. Overall, the dynamic mechanism consistently improves efficiency across all configurations. With dynamic injection, the CNF-solving time remains low, whereas the static baseline ranges from 0.15\,s to 427.56\,s. Dynamic injection reduces the average CNF-solving time from 61.99\,s to 0.65\,s. Notably, the benefit becomes more pronounced as the benchmark grows harder. For example, at \texttt{EdgeNum}=300 and \texttt{GroupNum}=4, the dynamic mechanism completes in 1.34--1.99\,s, while the static baseline requires 339.53--427.56\,s. These results indicate that dynamic injection effectively tightens constraints, thereby shrinking the search space and reducing solving time.

\begin{tcolorbox}[
  enhanced,
  colback=gray!5,
  colframe=gray!500,
  arc=2mm,
  boxrule=1.5pt,
  left=1.5mm,right=1.5mm,top=1mm,bottom=1mm
]
\textbf{Answer to RQ3: }
On real-world microservice benchmarks, enabling the dynamic injection mechanism reduces fault injection number by 68.47\%, CNF-solving time by 97.68\%, and end-to-end time by 77.13\%. This improvement stems from feedback-driven adaptation of the maximum number of APIs per combinatorial fault, which limits candidate solutions, shrinks the search space, and eliminates redundant injections and solver invocations.
\end{tcolorbox}

\subsection{RQ4: Effectiveness of Robustness Optimization}
\label{rq4:title}
\subsubsection{Baselines}
To the best of our knowledge, FastFI is the first work to use fault-injection results to guide the enhancement of API call-site robustness in microservices. To evaluate the effectiveness of our robustness optimization, we implemented a greedy heuristic baseline that maximizes per-request combinatorial-fault coverage under a budget constraint. For each request, we encoded all combinatorial faults as SAT constraints and built an inverted index mapping APIs to the combinatorial faults that contain them, enabling fast coverage computation. Given a budget \(B\), the heuristic first targets high-priority requests: it calculates the marginal coverage gain of each candidate API for the uncovered faults of high-priority requests, selects the API with the largest marginal gain, and iteratively updates coverage and budget. If the budget remains after fully covering the high-priority requests, the heuristic proceeds to low-priority requests and greedily selects APIs in descending order of frequency among the remaining faults to maximize additional coverage per selection. This heuristic provides a strong baseline for comparison.

\subsubsection{Metrics}
In this evaluation, we used the following four metrics: 

\noindent\textbullet\ \textbf{Identification Time:} the time required to identify the critical APIs under a given budget constraint.

\noindent\textbullet\ \textbf{\boldmath$CR_{b}$:} the coverage rate of combinatorial faults associated with low-priority requests under budget level $b$, given that we enforced 100\% coverage for all combinatorial faults on high-priority requests at the same budget level.

\noindent\textbullet\ \textbf{\boldmath$MCG_{b}$:} the marginal coverage gain per unit budget increase from level $b-1$ to $b$, defined as:
\begin{equation}
\label{math:mcg}
MCG_{b}
= \frac{CN_{b}-CN_{b-1}}
       {Budget_{b}-Budget_{b-1}}
\qquad (b>1),
\end{equation}
where $CN_{b}$ is the number of combinatorial faults associated with low-priority requests covered at budget level $b$ and $Budget_{b}$ is the cumulative budget consumed at that level.

\noindent\textbullet\ \textbf{\boldmath$AFVR_b$:} the average fault validity ratio after hardening the call-site robustness selected under budget level $b$. For each request, we re-injected all previously identified valid combinatorial faults and computed the fraction that remain failure-inducing after enhancing, defined as:
\begin{equation}
AFVR_b=\frac{1}{|\mathcal{R}|}\sum_{r\in\mathcal{R}}
\frac{\text{rvf}(\mathcal{F}_r)}
{|\mathcal{F}_r|},
\end{equation}
where $\mathcal{R}$ denotes the set of the evaluated requests, $\mathcal{F}_r$ denotes the set of valid combinatorial faults for request $r$ before enhancing, and $\text{rvf}(\mathcal{F}_r)$ denotes the number of combinatorial faults in $\mathcal{F}_r$ that remain valid and failure-inducing after enhancing. A smaller $AFVR_b$ indicates that fewer previously valid combinatorial faults remain failure-inducing after enhancing, and thus reflects more effective enhancing.

\subsubsection{Results}
\paragraph{(i) Results on Real-World Microservice Benchmarks.}

\begin{table*}[t]
\captionsetup{justification=centering}
\caption{$CR_{b}$, $MCG_{b}$, $AFVR_{b}$, and identification time under different budgets on real-world microservice benchmarks. \textit{Note:} OB-$n$/HR-$n$/SS-$n$/TT-$n$ denote Online Boutique/Hotel Reservation/Sock Shop/Train Ticket deployed with $n$ replicas per microservice.}
\label{tab:rq4-result}
\centering
\footnotesize
\begingroup
\setlength{\tabcolsep}{6pt}
\renewcommand{\arraystretch}{1.2}
{%
\begin{tabular}{c|l|c|cc|cc|cc|cc}
\toprule
\multirow{2}{*}{Benchmark} &
\multirow{2}{*}[-0.5ex]{\makecell{High-Priority\\Request}} &
\multirow{2}{*}{$Budget_{b}$} &
\multicolumn{2}{c|}{$CR_{b}$ (\%)} &
\multicolumn{2}{c|}{$MCG_{b}$} &
\multicolumn{2}{c|}{$AFVR_{b}$ (\%)} &
\multicolumn{2}{c}{Identification Time (ms)} \\
\cline{4-5}\cline{6-7}\cline{8-9}\cline{10-11}
\rule{0pt}{3ex}
& & & Heuristic & FastFI & Heuristic & FastFI & Heuristic & FastFI & Heuristic & FastFI \\
\midrule
\multirow{7}{*}{OB-6} & \multirow{7}{*}{GET /} & 4  & 55.00 & 55.00 & -- & -- & 33.33 & 33.33 & \textbf{12.5} & 23.1 \\
& & 5  & 60.00 & \textbf{65.00} & 1 & \textbf{2} & 29.76 & \textbf{25.60} & \textbf{14.9} & 39.2 \\
& & 6  & 70.00 & \textbf{75.00} & 2 & 2 & 22.62 & \textbf{17.86} & \textbf{14.8} & 39.7 \\
& & 7  & 75.00 & \textbf{85.00} & 1 & \textbf{2} & 19.05 & \textbf{10.71} & \textbf{21.2} & 51.8 \\
& & 8  & 80.00 & \textbf{90.00} & 1 & 1 & 15.48 & \textbf{7.14}  & \textbf{18.7} & 41.0 \\
& & 9  & 90.00 & \textbf{95.00} & \textbf{2} & 1 & 7.74  & \textbf{3.57}  & \textbf{17.3} & 41.3 \\
& & 10 & 100.00 & 100.00 & \textbf{2} & 1 & 0.00 & 0.00 & \textbf{16.5} & 41.5 \\
\midrule
\multirow{7}{*}{HR-6} & \multirow{7}{*}{\makecell{GET /recom-\\mendations}} & 2 & 14.29 & 14.29 & -- & -- & 60.00 & 60.00 & \textbf{4.8} & 6.6 \\
& & 3 & 28.57 & 28.57 & 1 & 1 & 53.33 & 53.33 & \textbf{6.1} & 16.3 \\
& & 4 & 42.86 & 42.86 & 1 & 1 & 46.67 & \textbf{26.67} & \textbf{6.8} & 17.9 \\
& & 5 & 57.14 & 57.14 & 1 & 1 & 40.00 & \textbf{20.00} & \textbf{6.5} & 18.7 \\
& & 6 & 71.43 & 71.43 & 1 & 1 & 33.33 & \textbf{13.33} & \textbf{6.1} & 18.9 \\
& & 7 & 85.71 & 85.71 & 1 & 1 & 16.67 & \textbf{6.67}  & \textbf{7.4} & 20.1 \\
& & 8 & 100.00 & 100.00 & 1 & 1 & 0.00 & 0.00 & \textbf{7.3} & 20.0 \\
\midrule
\multirow{10}{*}{SS-6} & \multirow{10}{*}{POST /cart} & 2  & 10.00 & 10.00 & -- & -- & 62.96 & 62.96 & \textbf{6.3} & 8.3 \\
& & 3  & 20.00 & 20.00 & 1 & 1 & 59.26 & \textbf{29.63} & \textbf{9.6} & 20.3 \\
& & 4  & 30.00 & 30.00 & 1 & 1 & 55.56 & \textbf{25.93} & \textbf{9.0} & 19.8 \\
& & 5  & 40.00 & 40.00 & 1 & 1 & 51.85 & \textbf{22.22} & \textbf{9.0} & 20.5 \\
& & 6  & 50.00 & 50.00 & 1 & 1 & 48.15 & \textbf{18.52} & \textbf{9.7} & 22.4 \\
& & 7  & 60.00 & 60.00 & 1 & 1 & 44.44 & \textbf{14.81} & \textbf{9.6} & 23.9 \\
& & 8  & 70.00 & 70.00 & 1 & 1 & 40.74 & \textbf{11.11} & \textbf{9.8} & 23.6 \\
& & 9  & 80.00 & 80.00 & 1 & 1 & 37.04 & \textbf{7.41}  & \textbf{9.6} & 23.8 \\
& & 10 & 90.00 & 90.00 & 1 & 1 & 33.33 & \textbf{3.70}  & \textbf{9.7} & 23.5 \\
& & 11 & 100.00 & 100.00 & 1 & 1 & 0.00 & 0.00 & \textbf{9.4} & 23.8 \\
\midrule
\multirow{17}{*}{TT-3} & \multirow{17}{*}{GET /food} & 6  & 16.30 & 16.30 & -- & -- & 69.40 & 69.40 & \textbf{16.6} & 30.6 \\
& & 7  & 17.39 & \textbf{21.74} & 1 & \textbf{5} & 66.36 & \textbf{64.76} & \textbf{22.8} & 51.5 \\
& & 9  & 25.00 & \textbf{32.61} & 3.5 & \textbf{5} & 57.07 & \textbf{55.47} & \textbf{23.7} & 53.6 \\
& & 11 & 35.87 & \textbf{43.48} & 5 & 5 & 47.78 & \textbf{46.18} & \textbf{22.9} & 54.4 \\
& & 13 & 42.39 & \textbf{54.35} & 3 & \textbf{5} & 38.49 & \textbf{36.89} & \textbf{23.1} & 53.0 \\
& & 15 & 48.91 & \textbf{65.22} & 3 & \textbf{5} & 29.21 & \textbf{27.60} & \textbf{23.5} & 53.5 \\
& & 17 & 55.43 & \textbf{73.91} & 3 & \textbf{4} & 21.53 & \textbf{19.92} & \textbf{22.9} & 52.7 \\
& & 19 & 61.96 & \textbf{79.35} & \textbf{3} & 2.5 & 17.14 & \textbf{15.34} & \textbf{23.2} & 54.5 \\
& & 21 & 64.13 & \textbf{82.61} & 1 & \textbf{1.5} & 14.62 & \textbf{13.01} & \textbf{23.1} & 52.9 \\
& & 23 & 66.30 & \textbf{84.78} & 1 & 1 & 12.14 & \textbf{11.01} & \textbf{22.9} & 53.2 \\
& & 25 & 77.17 & \textbf{86.96} & \textbf{5} & 1 & 9.88 & \textbf{9.15} & \textbf{23.3} & 54.4 \\
& & 27 & 85.87 & \textbf{89.13} & \textbf{4} & 1 & 7.40 & \textbf{7.25} & \textbf{25.2} & 54.1 \\
& & 29 & 89.13 & \textbf{91.30} & \textbf{1.5} & 1 & 5.80 & 5.80 & \textbf{26.7} & 54.8 \\
& & 31 & 93.48 & 93.48 & \textbf{2} & 1 & 4.35 & 4.35 & \textbf{27.1} & 52.8 \\
& & 33 & 95.65 & 95.65 & 1 & 1 & 2.90 & 2.90 & \textbf{28.3} & 52.6 \\
& & 35 & 97.83 & 97.83 & 1 & 1 & 1.45 & 1.45 & \textbf{30.5} & 53.8 \\
& & 37 & 100.00 & 100.00 & 1 & 1 & 0.00 & 0.00 & \textbf{30.8} & 54.9 \\
\bottomrule
\end{tabular}}%
\endgroup
\end{table*}

As summarized in Table~\ref{tab:rq4-result}, we compared the two-stage FastFI method with the heuristic algorithm across four metrics. To preserve realism, we assigned request priorities based on the frequency of observed requests within a fixed window, with OB-6 using \texttt{GET /}, HR-6 using \texttt{GET /recommendation}, SS-6 using \texttt{POST /carts}, and TT-3 using \texttt{GET /food} as the high-priority requests. The results show that on HR-6 and SS-6, both methods achieve the same values for the metrics $CR_{b}$ and $MCG_{b}$ because the two benchmarks lack common API invocations across requests, which limits FastFI's global optimization advantage. This also reflects the limitation of the benchmarks themselves, which ignore the API overlap among requests seen in production environments (see Table~\ref{tab:Alibaba-Dataset-Result-api-service}). In contrast, OB-6 and TT-3 expose a substantial set of APIs shared across requests. Consequently, FastFI attains higher initial API coverage and yields an \(MCG_{b}\) value with the expected long-tail pattern---large early gains that gradually taper and stabilize. As the budget increases and becomes effectively unconstrained, both methods eventually reach 100\% coverage; however, under realistic (budget-constrained) settings, FastFI achieves higher coverage than the heuristic baseline in most cases. Moreover, the heuristic's \(MCG_{b}\) exhibits noticeable fluctuations, indicating poorer stability.

Regarding $AFVR_b$, these results provide closed-loop validation of our recommendations: for each budget, we hardened the call-site robustness of the APIs selected by the given method, re-injected all previously identified valid combinatorial faults for each request, and then computed $AFVR_b$. FastFI consistently yields lower values than the heuristic baseline under most budget settings. In particular, at small budgets, FastFI and the heuristic cover all combinatorial faults associated with high-priority requests, resulting in identical \(AFVR_{b}\). As the budget increases, the gap between the two methods gradually widens; when the budget becomes effectively unconstrained, the gap narrows again and \(AFVR_{b}\) converges to the same value for both methods. These results indicate that hardening the call-site robustness of APIs identified by FastFI invalidates a larger fraction of previously valid combinatorial faults under limited constraints, indicating more effective hardening and higher system reliability than heuristic-guided hardening.
In terms of identification time, FastFI incurs a higher cost than the heuristic because it solves a more complex Partial Max-SAT problem in pursuit of global optimality.

\paragraph{(ii) Results on the Simulated Microservice Benchmark.}
\begin{figure}[t]
\captionsetup{skip=2pt}
  \centering
  \includegraphics[width=\linewidth]{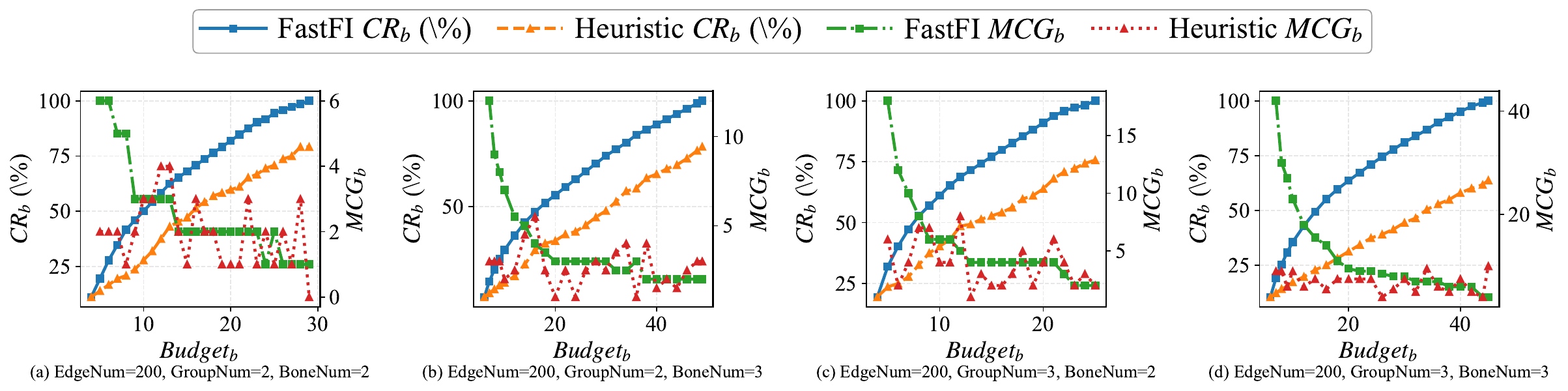}
  \caption{$CR_{b}$ and $MCG_{b}$ under different budgets on the simulated microservice benchmark.}
  \label{fig:RQ4-1}
\end{figure}
\begin{figure}[t]
\captionsetup{skip=2pt}
  \centering
  \includegraphics[width=\linewidth]{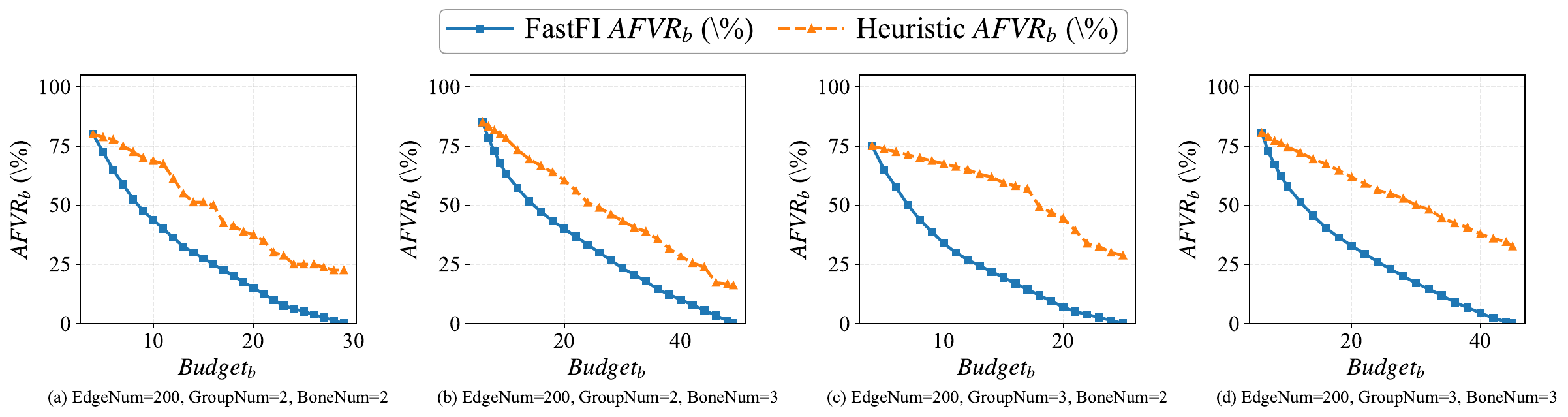}
  \caption{$AFVR_{b}$ under different budgets on the simulated microservice benchmark.}
  \label{fig:RQ4-2}
\end{figure}
\begin{figure}[t]
\captionsetup{skip=2pt}
  \centering
  \includegraphics[width=\linewidth]{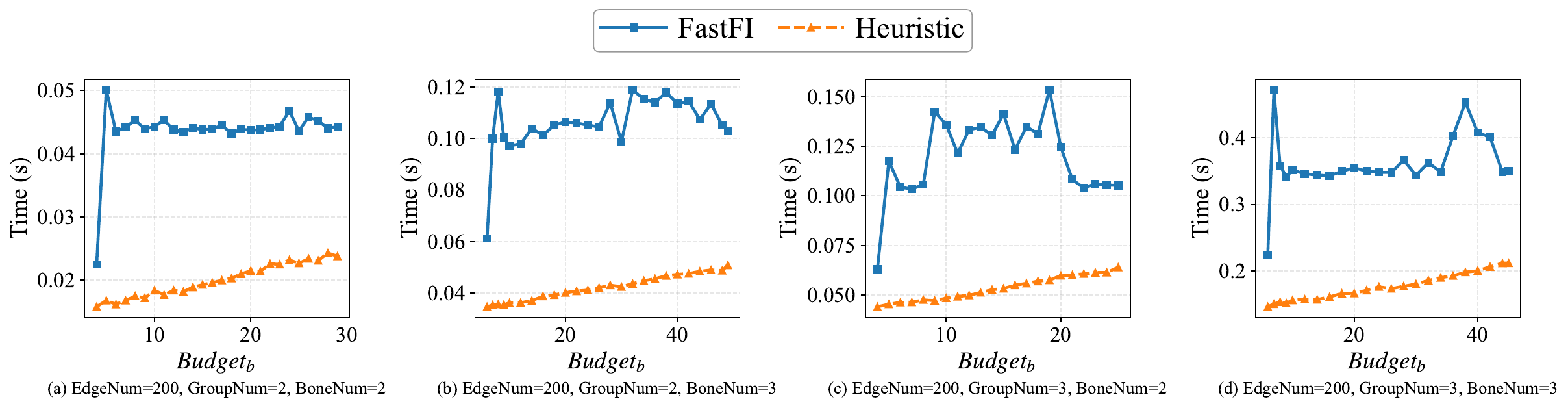}
  \caption{Identification time under different budgets on the simulated microservice benchmark.}
  \label{fig:RQ4-3}
\end{figure}
Fig.~\ref{fig:RQ4-1} shows the $CR_{b}$ and $MCG_{b}$ results on the large-scale benchmark. $CR_{b}$ of FastFI grows and approaches $100\%$ as the budget increases, whereas the heuristic remains below FastFI at all budget levels. $MCG_{b}$ of FastFI exhibits a smooth long-tail pattern: high early gains that taper and stabilize, showing effective prioritization of high-impact APIs and steady convergence. However, the heuristic’s $MCG_{b}$ shows pronounced fluctuations and remains consistently lower, indicating a less stable search.

As shown in Fig.~\ref{fig:RQ4-2}, FastFI consistently attains a lower \(AFVR_{b}\) than the heuristic baseline under most budget settings across all configurations. Moreover, the figure makes the trend explicit: as the budget increases, the \(AFVR_{b}\) gap between the two methods initially widens; as the budget continues to grow and becomes effectively unconstrained, the gap gradually narrows. This convergence occurs because, with an unlimited budget, both methods can cover all combinatorial faults. This advantage remains stable across configuration variations, indicating that FastFI’s robustness optimization strategy generalizes across different multipath structures and overlap regimes in the simulated benchmark. Overall, these results corroborate the real-world findings and provide additional evidence that hardening the call-site robustness of the APIs identified by FastFI improves system reliability under combinatorial faults with limited budgets.
As shown in Fig.~\ref{fig:RQ4-3}, FastFI incurs a $2.2\times$ identification time overhead over the heuristic due to solving a Partial Max-SAT problem, yet remains within an acceptable range.

\begin{tcolorbox}[
  enhanced,
  colback=gray!5,
  colframe=gray!500,
  arc=2mm,
  boxrule=1.5pt,
  left=1.5mm,right=1.5mm,top=1mm,bottom=1mm
]
\textbf{Answer to RQ4: }
FastFI's two-stage method stably identifies critical APIs to enhance call-site robustness, achieving higher early gains by prioritizing high-impact APIs, particularly when multiple requests invoke overlapping APIs.
Although FastFI solves a Partial Max-SAT problem and incurs a 2.20$\times$ identification time overhead over the heuristic, both complete within 0.5\,s, making the overhead acceptable. After enhancing the call-site robustness of the APIs identified by FastFI, we re-injected all discovered valid combinatorial faults for each request and measured the average fault validity ratio. The results show that FastFI-guided enhancement leads to higher system reliability.
\end{tcolorbox}

\subsection{RQ5: Overhead of FastFI}
\begin{table*}[t]
\centering
\captionsetup{justification=centering}
\caption{Comparison of CPU and memory usage across fault-injection algorithms for different requests. \textit{Note:} OB-$n$/HR-$n$/SS-$n$/TT-$n$ denote Online Boutique/Hotel Reservation/Sock Shop/Train Ticket deployed with $n$ replicas per microservice.}
\label{tab:cpu-mem-combined}
\footnotesize
\begingroup
\setlength{\tabcolsep}{8pt}
\renewcommand{\arraystretch}{1.3}
\begin{tabular}{c|l|cccc|cccc}
\toprule
\multirow{2}{*}[0ex]{Benchmark} &
\multirow{2}{*}[0ex]{Request} &
\multicolumn{4}{c|}{CPU Usage (\%)} &
\multicolumn{4}{c}{Memory Usage (MB)} \\
\cline{3-6}\cline{7-10}
\rule{0pt}{3ex}
& & LDFI & IntelliFI & MicroFI & FastFI & LDFI & IntelliFI & MicroFI & FastFI \\
\midrule

\multirow{4}{*}{OB-4} & GET /                & 1.05 & 1.03 & \textbf{0.87} & 0.91 & 138.3 & 138.7 & 139.1 & \textbf{130.6} \\
    & GET /product/\{id\}  & 1.11 & \textbf{1.02} & 1.12 & 1.18 & 138.1 & 138.9 & 139.2 & \textbf{130.1} \\
    & GET /cart            & 1.15 & 1.12 & 1.18 & \textbf{1.06} & 139.6 & 140.8 & 142.0 & \textbf{132.4} \\
    & POST /checkout       & 2.16 & 2.23 & 2.31 & \textbf{1.08} & 149.9 & 148.2 & 148.5 & \textbf{132.9} \\
\midrule

\multirow{4}{*}{OB-6} & GET /                & 1.21 & 1.08 & 1.18 & \textbf{1.06} & 149.5 & 140.7 & 141.2 & \textbf{132.9} \\
    & GET /product/\{id\}  & 2.82 & 2.43 & 3.33 & \textbf{1.15} & 161.9 & 141.4 & 139.6 & \textbf{130.5} \\
    & GET /cart            & 2.52 & 2.53 & 2.52 & \textbf{1.25} & 153.2 & 148.9 & 154.9 & \textbf{129.7} \\
    & POST /checkout       & 24.56 & 21.20 & 20.05 & \textbf{3.40} & 196.7 & 192.7 & 200.2 & \textbf{172.4} \\
\midrule

\multirow{3}{*}{HR-4} & GET /recommendations & 1.25 & 1.30 & \textbf{1.24} & 1.87 & 138.7 & 138.7 & 138.6 & \textbf{128.0} \\
    & GET /hotels          & 1.05 & \textbf{1.02} & 1.07 & 1.30 & 138.2 & 137.9 & 138.7 & \textbf{124.1} \\
    & POST /reservation    & 1.38 & 1.38 & \textbf{1.23} & 1.91 & 138.5 & 138.3 & 138.8 & \textbf{128.2} \\
\midrule

\multirow{3}{*}{HR-6} & GET /recommendations & \textbf{1.05} & 1.23 & 1.35 & 1.67 & 139.3 & 139.0 & 138.9 & \textbf{128.9} \\
    & GET /hotels          & \textbf{1.09} & 1.11 & 1.45 & 1.29 & 140.4 & 137.7 & 138.1 & \textbf{129.6} \\
    & POST /reservation    & \textbf{1.34} & 1.36 & 1.39 & 1.70 & 137.9 & 137.8 & 137.9 & \textbf{127.1} \\
\midrule

\multirow{3}{*}{SS-4} & GET /login           & \textbf{1.21} & 1.35 & 1.25 & 1.83 & 138.4 & 138.1 & 138.2 & \textbf{129.8} \\
    & POST /orders         & 1.76 & 1.82 & 1.74 & \textbf{1.29} & 144.8 & 138.6 & 139.2 & \textbf{130.7} \\
    & POST /cart           & 1.51 & 1.72 & \textbf{1.36} & 1.65 & 138.9 & 137.5 & 138.0 & \textbf{129.2} \\
\midrule

\multirow{3}{*}{SS-6} & GET /login           & 1.25 & 1.29 & \textbf{1.22} & 1.66 & 139.9 & 138.3 & 138.1 & \textbf{127.8} \\
    & POST /orders         & 16.24 & 16.56 & 15.41 & \textbf{2.08} & 146.1 & 152.7 & 145.6 & \textbf{133.8} \\
    & POST /cart           & 1.22 & 1.19 & \textbf{1.08} & 1.36 & 139.7 & 138.9 & 138.4 & \textbf{130.2} \\
\midrule

\multirow{6}{*}{TT-2} & POST /left       & 1.33 & 1.28 & 1.41 & \textbf{0.87} & 150.0 & 142.6 & 148.7 & \textbf{114.4} \\
    & POST /cheapest   & 2.25 & 1.94 & 2.11 & \textbf{1.56} & 198.4 & 169.7 & 178.5 & \textbf{125.4} \\
    & POST /minStation & 2.14 & 1.93 & 2.25 & \textbf{1.44} & 186.9 & 160.3 & 163.1 & \textbf{120.0} \\
    & POST /quickest   & 2.33 & 2.08 & 2.02 & \textbf{1.63} & 190.3 & 168.8 & 184.3 & \textbf{122.0} \\
    & GET /food        & 0.89 & \textbf{0.86} & 0.93 & 1.09 & 138.6 & 137.4 & 138.4 & \textbf{129.2} \\
    & POST /preserve   & 1.59 & 1.70 & 1.68 & \textbf{1.13} & 183.7 & 158.6 & 168.5 & \textbf{126.0} \\
\midrule

\multirow{6}{*}{TT-3} & POST /left       & 5.13 & 4.89 & 5.12 & \textbf{2.42} & 544.4 & 507.8 & 515.2 & \textbf{136.8} \\
    & POST /cheapest   & 16.11 & 14.76 & 14.72 & \textbf{5.00} & 3581.2 & 2769.8 & 2867.2 & \textbf{150.4} \\
    & POST /minStation & 16.33 & 14.19 & 14.26 & \textbf{4.86} & 2764.8 & 2252.8 & 2457.6 & \textbf{144.3} \\
    & POST /quickest   & 15.98 & 14.96 & 15.71 & \textbf{5.18} & 3483.6 & 2972.6 & 3174.4 & \textbf{149.7} \\
    & GET /food        & 0.82 & \textbf{0.64} & 0.81 & 0.78 & 140.7 & 139.9 & 141.1 & \textbf{130.2} \\
    & POST /preserve   & 5.07 & 4.72 & 4.87 & \textbf{2.33} & 1537.8 & 1274.4 & 1344.9 & \textbf{144.1} \\
\bottomrule
\end{tabular}
\endgroup
\end{table*}

To quantify FastFI’s overhead, we measure resource utilization. We ran FastFI and all baselines on identical hardware and deployments, monitored complete evaluation runs, and recorded average CPU and memory usage during fault injection.

Table~\ref{tab:cpu-mem-combined} reports average CPU and memory usage per request on real-world microservice benchmarks. FastFI exhibits low CPU usage (1.84\% on average). In comparison, MicroFI, IntelliFI, and LDFI exhibit average CPU usage of 4.01\%, 4.00\%, and 4.28\%, respectively. The gap is even larger for more complex requests, such as the OB-6 request (\texttt{POST /checkout}), the SS-6 request (\texttt{POST /cart}), and the TT-3 requests (\texttt{POST /cheapest}, \texttt{POST /minStation}, and \texttt{POST /quickest}). The baselines require 16.74\% CPU usage on average on these requests, whereas FastFI uses only 4.10\%, reducing CPU usage by 75.48\%. The gap arises because CNF solving dominates the runtime, and the baselines rely on a general-purpose SAT solver, Z3, which is CPU-heavy for our instances.

In terms of memory usage, FastFI consumes 132.2 MB on average, whereas MicroFI, IntelliFI, and LDFI require 448.6 MB, 428.7 MB, and 500.3 MB, respectively.
This reduction mainly stems from FastFI discarding stale candidates derived from older CNF formulas and injecting only candidates derived from the latest CNF formula.
In contrast, the baselines retain candidates that satisfy earlier CNF formulas, thereby accumulating additional memory overhead.
The effect is particularly pronounced for TT-3 requests (\texttt{POST /cheapest}, \texttt{POST /minStation}, \texttt{POST /quickest}, and \texttt{POST /preserve}): FastFI uses only 147.1 MB on average, while the baselines rise to 2,540.1 MB, which can be impractical in practice.


\begin{tcolorbox}[
  enhanced,
  colback=gray!5,
  colframe=gray!500,
  arc=2mm,
  boxrule=1.5pt,
  left=1.5mm,right=1.5mm,top=1mm,bottom=1mm
]
\textbf{Answer to RQ5: }
FastFI substantially reduces resource overhead on real-world microservice benchmarks. It lowers CPU usage by 54.00\%, 53.98\%, and 56.91\% compared with MicroFI, IntelliFI, and LDFI, respectively, and reduces memory usage by 70.52\%, 69.16\%, and 73.57\%. These results suggest that FastFI is practical for deployment, due to the low-cost DFS-based solver and pruning strategy that focuses on solutions from the latest CNF formula while discarding stale candidates.
\end{tcolorbox}

\subsection{Case Study}
To demonstrate FastFI’s practicality, we conducted a case study on the Online Boutique benchmark deployed with four replicas per microservice (Fig.~\ref{fig:4.7}). We first configured the fault-injection experiment by specifying the target request endpoints for injection, the communication protocol between services, and an upper bound for our $k$-fault-based dynamic injection, which limits the maximum number of injection points in each candidate fault set.
 In the benchmark, each valid combinatorial fault corresponds to a replica combination of the same downstream API: it includes four injection points, one targeting the API endpoint on each replica, rather than combinations involving multiple distinct APIs. This pattern indicates that the call site lacks sufficient error handling and fallback logic when all replicas of a downstream API become unavailable.

After deriving the fault list via fault injection, we selected APIs to enhance their call-site robustness within a fixed budget. We profiled request frequencies over a time window to assign request priorities: we treated \texttt{GET /} as high priority and \texttt{GET /product/\{id\}}, \texttt{GET /cart}, and \texttt{POST /checkout} as low priority. For high-priority requests, we required full coverage of all combinatorial faults in their lists to preserve availability under faults. Given a total hardening budget of five APIs, four budget slots are allocated to the high-priority request, and the remaining budget is allocated to maximize coverage of faults in low-priority requests. If multiple APIs offered equivalent coverage, we broke ties using \emph{APIRank}~\cite{MicroFI}, which ranks APIs based on call frequency and call-graph topology.

\begin{figure}[t]
\captionsetup{skip=2pt}
  \centering
  \includegraphics[width=\linewidth]{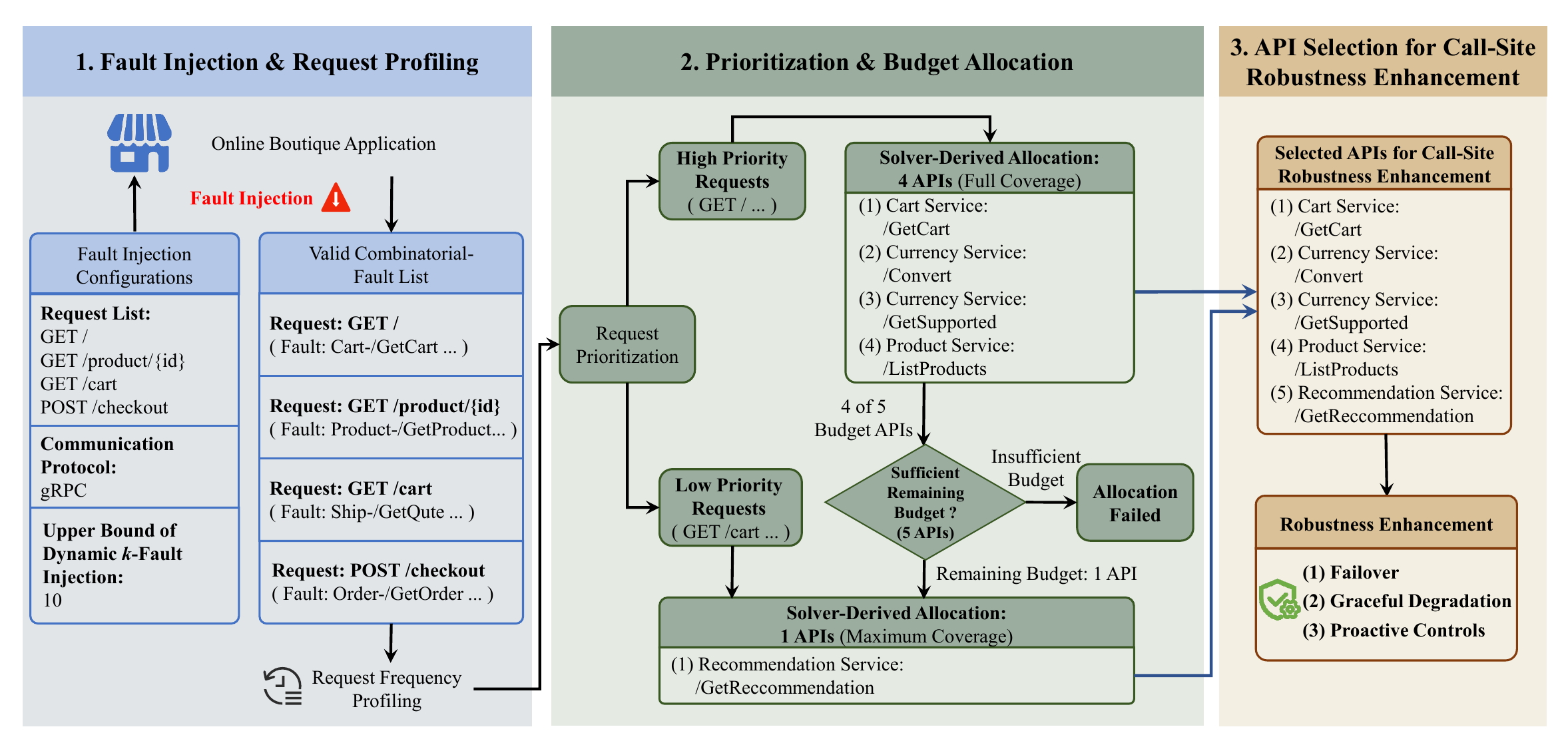}
  \caption{Case study of applying FastFI to the Online Boutique application.}
  \label{fig:4.7}
\end{figure}

The output is a list of APIs whose call-site robustness should be enhanced. Enhancing the robustness of their call sites (e.g., via failover, graceful degradation, or proactive controls) yields the largest improvement in system reliability under a constrained budget. Crucially, FastFI operates in a black-box manner, pinpointing critical call sites without requiring access to source code, thereby reducing analysis complexity and improving applicability across systems.

\section{Discussion}
\label{sec:5}
\subsection{Limitations}
FastFI requires target applications to implement OpenTelemetry instrumentation so that tracestate can be propagated along service call chains, enabling request-level fault injection. Currently, fault injection is implemented via Istio's EnvoyFilter, which configures Envoy dynamically to intercept and mutate network traffic. However, the sidecar-based service mesh incurs measurable overhead, including additional memory and CPU consumption. According to our monitoring of OB-6 with over 60 pods, each sidecar container in a pod averaged about 100 MB of memory and 150 millicores of CPU under request load. Additionally, the configuration of the Istio-proxy container shows that it requests 128 MB of memory and 100 millicores of CPU. To overcome these limitations and achieve truly non-intrusive fault injection, we plan to leverage eBPF~\cite{eBPF} for kernel-level traffic interception and to integrate ZeroTrace~\cite{ZeroTracer}, an eBPF-based, instrumentation-free tracing tool, to inject metadata directly into packets in kernel space. This design eliminates the need for a service mesh or any application-layer code changes, enabling request-level fault injection with no application-level instrumentation and minimal performance impact.

\subsection{Threats to Validity}
\paragraph{Internal Threats to Validity.} Internal validity threats primarily concern controllability and reproducibility during experimentation. Because the competing microservice fault-injection methods LDFI~\cite{LDFI-Netflix}, IntelliFI~\cite{IntelliFI}, and MicroFI~\cite{MicroFI} do not publicly release complete implementations or technical details, we mitigated this threat by carefully analyzing each baseline’s core algorithmic description, system architecture, and critical implementation details, and by re-implementing them faithfully according to the specifications in their papers. During re-implementation, we paid particular attention to preserving the baselines’ core logic and optimization strategies to ensure a fair and credible comparison.

\paragraph{External Threats to Validity.} The principal external validity threat arises from the gap between academic/open-source microservice benchmarks and production-grade systems, particularly in how redundant execution paths and reliability mechanisms are realized in practice. Most widely used open-source microservice benchmarks lack business-level fallback and degradation logic (e.g., failover, graceful degradation, or proactive controls), so injecting a single fault often causes a request to fail outright. As a result, alternative execution paths are commonly approximated via replica-based, instance-level failover, which may increase clause overlap and inflate the number of paths compared to production settings. To mitigate this threat, we designed a multi-level evaluation. We first evaluated FastFI on standard open-source benchmarks, using the same replica-based construction adopted by prior work, such as IntelliFI and MicroFI, to ensure comparability under a high-overlap regime. We then constructed a large-scale simulated microservice benchmark designed to approximate production-grade, high-reliability practices by explicitly controlling the multipath structure and overlap through scaling service counts, dependency patterns, and call-frequency distributions. Finally, we reported overlap statistics $ACO$ and swept key structural parameters to demonstrate that FastFI’s conclusions hold across a wide range of overlap regimes and system configurations, thereby strengthening the generalizability of our findings beyond any single benchmark.

\paragraph{Construct Threats to Validity.}
A construct validity threat is whether $AFVR_b$, defined in RQ4, faithfully captures the intended notion of reliability improvement after hardening. By definition, we compute $AFVR_b$ as the fraction of previously identified valid combinatorial faults that remain failure-inducing after hardening, and then average it across requests. Although this metric directly reflects how effectively a hardening strategy eliminates known failure-inducing combinations, it remains a proxy for production reliability.
Specifically, $AFVR_b$ treats all faults equally and does not account for failure severity, user impact, or occurrence frequency; thus, invalidating rare faults may improve $AFVR_b$ without yielding proportional operational benefits. Moreover, $AFVR_b$ focuses solely on failure/non-failure outcomes and does not capture potential side effects of hardening (e.g., semantic degradation of responses or latency overhead), which can be important in practice.
To mitigate these threats, we treat $AFVR_b$ as a targeted reliability proxy that measures how effectively hardening invalidates previously observed failure-inducing combinations under a fixed injection model. We evaluated $AFVR_b$ across multiple benchmarks and budgets, explicitly documented the injection protocol and success criteria, and applied the same hardening strategy across all methods. These controls reduce measurement and comparison artifacts, so observed differences in $AFVR_b$ primarily reflect the quality of call-site selection rather than experimental tuning.

\section{Related Work}
\label{sec:6}
\textbf{Fault Injection.} Fault injection (FI) proactively surfaces latent system flaws, and chaos engineering provides practical guidance for conducting FI in production environments~\cite{ChaosEngineeringBasic}. Prior work has improved FI by developing more effective scenario-selection and analysis techniques~\cite{Coverage-ICC, FailureModeAnalysis, CausalFaultInject, RiskAnalysis}. Distributed-systems research has proposed diverse FI techniques to improve testing realism and experimental reproducibility~\cite{SyscallChaos, Rainmaker, FeedbackDrivenFailureReproduction, SlowFault}. In microservices, LDFI~\cite{LDFI-2015} formulates FI as a SAT problem to uncover minimal combinatorial faults, and it has been deployed in production at Netflix~\cite{LDFI-Netflix, Netflix-ICSE}. IntelliFI~\cite{IntelliFI} accelerates discovery via fitness-guided search; MicroFI~\cite{MicroFI} enables request-level injection with APIRank-based selection under a bounded injection budget; and MicroRes~\cite{MicroRes} leverages FI to profile reliability and localize bottlenecks. Production toolchains ~\cite{ChaosMesh, ChaosBlade, ChaosMonkey} support a broad range of fault types; Gremlin~\cite{Gremlin} adopts a proxy-based FI framework; 3MileBeach~\cite{3MileBeach} enables request-level injection through trace data; and Filibuster~\cite{Filibuster} integrates service-level FI with functional testing to verify recovery mechanisms. The reliability of cluster controllers and resilience under unreliable databases have also been studied~\cite{ControllerFault, DatabaseFault}. However, prior studies provide limited guidance on how to systematically improve API call-site robustness based on FI results.

\textbf{Robustness of API Call Sites.} Robust API call sites are critical to both end-user experience and system reliability in large-scale systems. Tail latency can dominate end-to-end response time, which motivates timeouts, retries, and hedged requests~\cite{Tail}. However, indiscriminate hedging can increase load and must therefore be load-aware and cost-bounded~\cite{Hedge}. In the presence of slow failures (e.g., network or I/O latency), fixed-threshold timeout and retry policies are brittle, whereas adaptive detection and mitigation mitigate performance degradation~\cite{SlowFault}. Moreover, error-handling code is often defect-prone, which can undermine recovery even when such policies are carefully tuned~\cite{ErrorHandlingCodeQuality}. Robustness, therefore, requires a coordinated approach including load-aware client invocation, adaptive server-side reliability, and systematic verification of error-handling paths.

\section{Conclusion}
\label{sec:7}
This paper presents FastFI, an efficient fault-injection framework for microservices that identifies valid combinatorial faults and recommends a budget-bounded set of APIs to enhance the call-site robustness. We design a DFS-based solver that rapidly enumerates all minimal SAT solutions and integrate it with a dynamic $k$-fault combinatorial injection pipeline to efficiently localize valid combinatorial faults. To maximize utility under tight DevOps timelines, we formulate API selection as a Partial Max-SAT problem, selecting a set of APIs subject to a budget and maximizing coverage of observed valid combinatorial faults. Extensive experiments show that FastFI outperforms existing methods in failure-discovery efficiency, runtime overhead, and the accuracy of identifying high-impact APIs. Moreover, enhancing the call-site robustness of selected APIs effectively improves system reliability.

In the future, we plan to develop a meshless, eBPF-based injection design that eliminates Istio sidecar overhead to achieve true zero instrumentation and to explore adaptive tuning of FastFI in cloud-native environments to better support continuous integration and deployment.

\section{Data Availability}
Our code and benchmarks are open-source and publicly available at the repository~\cite{FastFI}.

\begin{acks}
This work is supported by the National Natural Science Foundation of China (No. 62572362) and the Shenzhen Science and Technology Program (No. CJGJZD20230724091659002).
\end{acks}

\bibliographystyle{ACM-Reference-Format}
\bibliography{Reference}

\end{document}